\documentclass[journal]{IEEEtran}
\usepackage{amsmath, amssymb,amsthm,cite}
\usepackage[dvips]{graphics}
\usepackage[utf8x]{inputenc}
\usepackage{graphicx}
\usepackage{psfrag}
\usepackage{bbm}
\usepackage{subfigure}
\usepackage{dblfloatfix}
\usepackage{mathtools}
\usepackage{algpseudocode}
\usepackage{algorithm}
\usepackage{hyperref}
\usepackage{color}
\usepackage{tabularx}
\usepackage{tcolorbox}
\usepackage{etoolbox}
\usepackage{mleftright}

\DeclareMathOperator*{\nn}{\nonumber}

\allowdisplaybreaks
\newcommand{\RNum}[1]{\uppercase\expandafter{\romannumeral #1\relax}}

\newtheorem{theorem}{Theorem}

\theoremstyle{definition}
\newtheorem{remark}{Remark}

\makeatletter
\def\blfootnote{\gdef\@thefnmark{}\@footnotetext}
\makeatother

\def\cY{{\mathcal Y}}
\def\cS{{\mathcal S}}

\def\cQ{{\mathcal Q}}

\newcommand{\pr}[1]{\left(#1\right)}

\title{Capacity of Finite-State Channels with Delayed Feedback}
\author{Bashar~Huleihel,~\IEEEmembership{Student~Member,~IEEE,}
        Oron~Sabag,~\IEEEmembership{Member,~IEEE,}
        Haim~H.~Permuter,~\IEEEmembership{Senior~Member,~IEEE,}
        and~Victoria~Kostina,~\IEEEmembership{Senior~Member,~IEEE}}
\begin{document}
\maketitle

\begin{abstract}
In this paper, we investigate the capacity of finite-state channels (FSCs) in the presence of delayed feedback. We show that the capacity of a FSC with delayed feedback can be computed as that of a new FSC with instantaneous feedback and an extended state. Consequently, graph-based methods to obtain computable upper and lower bounds on the delayed feedback capacity of unifilar FSCs are proposed. Based on these methods, we establish that the capacity of the trapdoor channel with delayed feedback of two time instances is given by $\log_2\left(\frac{3}{2}\right)$. In addition, we derive an analytical upper bound on the delayed feedback capacity of the binary symmetric channel with a no consecutive ones input constraint. This bound also serves as a novel upper bound on its non-feedback capacity, which outperforms all previously known bounds. Lastly, we demonstrate that feedback does improve the capacity of the dicode erasure channel.
\end{abstract}

\blfootnote{This work was supported in part by the German Research Foundation (DFG) via the German Israeli Project Cooperation (DIP), the Israel Science Foundation (ISF) under Grant 899/21, the NSF-BSF grant, and the Israeli Innovation Authority. The work of O. Sabag was supported in part by the Israel Science Foundation (ISF) under Grant 1096/23. The work of V. Kostina was supported in part by National Science Foundation (NSF), under Grant CCF-1751356 and Grant CCF1956386. This paper was presented in part at the 2022 IEEE International Symposium on Information Theory \cite{Huleihel_trapdoor}.

B. Huleihel and H. H. Permuter are with the Department
of Electrical and Computer Engineering, Ben-Gurion University of the
Negev, Beer-Sheva 84105, Israel (e-mail: basharh@post.bgu.ac.il; haimp@post.bgu.ac.il).

O. Sabag is with the Rachel and Selim Benin School of Computer Science and Engineering, Hebrew University of Jerusalem, Jerusalem 91904, Israel (e-mail:  oron.sabag@mail.huji.ac.il).

V. Kostina is with the Department of Electrical Engineering, California Institute of Technology, Pasadena, CA 91125 USA (vkostina@caltech.edu).}


\section{Introduction}\label{sec:intro}
A finite-state channel (FSC) is a widely used statistical model for a channel with memory \cite{McMillan1953TheBT,Shannon_FSC,Blackwell58}. The memory of this channel is represented by an underlying channel state that takes values from a finite set. This model has been used in many practical applications, including wireless communication \cite{FSC_Wirless1,FSC_Wirless2,Wang95_FSC_usful_for_radiochannels}, molecular communication \cite{MolecFSCTransComm,MolecularSurvey}, and magnetic recording \cite{FSC_Magnetic}. An example of its versatility is the ability to model a memoryless channel with an input constraint by introducing a dummy sink state whose capacity is zero in case the constraint is violated. Generally speaking, the capacity formula of a FSC, whether or not feedback is allowed, is given by a multi-letter expression which is hard to evaluate. The main focus of this paper is on an important class of FSCs, known as unifilar FSCs. For these channels, the new channel state is determined by a time-invariant function of the previous channel state, the current channel input, and the current channel output.

The capacity of a unifilar FSC with instantaneous feedback has been broadly investigated in the literature, while instantaneous feedback refers to the case where at time $t$, the encoder has access to the channel outputs up to time $t-1$. This has resulted in several powerful methodologies that have been employed to derive simple capacity expressions and optimal coding schemes for well-known instances of unifilar FSCs with feedback \cite{Chen05,PermuterCuffVanRoyWeissman08,Ising_channel,Sabag_BEC,trapdoor_generalized,Ising_artyom_IT,Shemuel_NOST,PeledSabag,Sabag_BIBO,AharoniSabagRL}.
We mention here three essential works that will be utilized in the current paper: in \cite{Sabag_UB_IT,OronBasharfeedback}, for a given \textit{$Q$-graph}\footnote{The $Q$-graph is an auxiliary directed graph that is used to map channel output sequences onto one of the graph nodes (e.g. see Fig. \ref{fig:1Markov}).}, single-letter upper and lower bounds on the feedback capacity of unifilar FSCs were introduced, as well as a methodology to evaluate the bounds; in \cite{Sabag_DB_FB}, an alternative methodology to derive computable capacity upper bounds was proposed. In particular, in \cite{Sabag_DB_FB} it was shown that the dual capacity upper bound can be formulated as a simple Markov decision process (MDP) with the MDP states, actions, and disturbances taking values within finite sets. The main advantage of this methodology compared to the single-letter $Q$-graph upper bound lies in the simplicity of deriving analytical upper bounds. The bounds, however, depend on the choice of a test distribution over the channel outputs ensemble. Therefore, in order for these bounds to be meaningful, the test distribution must be chosen carefully.

\begin{figure}[t]
\centering
    \psfrag{E}[][][1]{Encoder}
    \psfrag{D}[][][1]{Decoder}
    \psfrag{C}[][][0.85]{$P(s_t,y_t|x_t,s_{t-1})$}
    \psfrag{V}[][][.9]{Delay $d$}
    \psfrag{M}[][][1]{$m$}
    \psfrag{Y}[][][1]{$y_t$}
    \psfrag{O}[][][1]{$\hat{m}$}
    \psfrag{Z}[][][1]{$y_{t-d}$}
    \psfrag{X}[][][1]{$x_t$}
    \includegraphics[scale = 0.4]{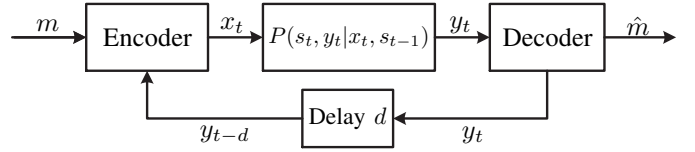}
    \caption{Finite-state channel with delayed feedback of $d$ time instances.}
    \label{fig:FSC}
\end{figure}
In this paper, we investigate the delayed feedback capacity of FSCs. This means that, in the case of a delay of $d$ time instances, the encoder has access to the channel outputs up to time $t-d$, in contrast to the standard feedback definition of access up to time $t-1$ (see Fig. \ref{fig:FSC}). We demonstrate that the capacity of a general FSC with delayed feedback can be computed as the capacity of a new reformulated FSC with instantaneous feedback. Specifically, we define the new channel state as $\tilde{S}_{t-1}=\left(S_{t-d},X_{t-d+1}^{t-1}\right)$, the new channel output as $\tilde{Y}_t=Y_{t-d+1}$, and leave the channel input unchanged, i.e., $\tilde{X}_t=X_t$. We prove that this new channel is an FSC and that its capacity with instantaneous feedback is equal to the capacity of the original FSC with delayed feedback of $d$ time instances. In addition, we show that if the original channel is a unifilar FSC, then the reformulated FSC is a unifilar FSC as well. As a result, for any unifilar FSC the methodologies from \cite{Sabag_UB_IT,OronBasharfeedback,Sabag_DB_FB} can also be applied when the feedback is delayed, by redefining the FSC. It is a known fact that feedback can only increase the feedforward capacity (i.e. the capacity without feedback). Accordingly, while studying the delayed feedback capacity is an important task in itself, it is also beneficial for deriving upper bounds on the feedforward capacity.

A major contribution of this paper is our investigation of the delayed feedback capacity of several well-known FSCs. We provide novel results regarding both their delayed feedback capacity and their feedforward capacity. Specifically, our first main result concerns the capacity of the trapdoor channel \cite{Blackwell_trapdoor}. Despite the extensive research efforts dedicated to the trapdoor channel, e.g. \cite{Ahl_kaspi87,Ahlswede98,PermuterCuffVanRoyWeissman08,kobayashi20033,Trapdoor_Lutz,Huleihel_Sabag_DB}, its feedforward capacity has remained an open problem for over sixty years.  In \cite{PermuterCuffVanRoyWeissman08}, it was shown that the feedback capacity of the trapdoor channel is equal to $\mathrm{C}^{\mathrm{fb}}_\mathrm{1} = \log_2\left(\frac{1+\sqrt{5}}{2}\right)\approx 0.6942$. 
In this paper, we consider the trapdoor channel with delayed feedback of two time instances. We show that the capacity in this scenario is equal to $\mathrm{C}^{\mathrm{fb}}_\mathrm{2}=\log_2\left(\frac{3}{2}\right)\approx 0.5850$, which is much closer to the lower bound on the feedforward capacity of $0.572$ \cite{kobayashi20033}. Further, by investigating a greater delay of the feedback, we provide a new upper bound on its feedforward capacity, which is approximately equal to $0.5765$. Next, we study the capacity of the binary symmetric channel (BSC) in the case where the input sequence is not allowed to contain two consecutive ones. For this setting, we derive an analytical upper bound on its capacity with delayed feedback of time instances that also serves as a novel upper bound on the feedforward capacity. It is interesting that for both the trapdoor channel and the input constrained BSC the capacity significantly degrades when only the previous channel output is not provided to the encoder. Finally, we demonstrate that the feedback capacity of the dicode erasure channel (DEC) \cite{PfitserLDPC_memory_erasure, henry_dissertation} is not equal to its feedforward capacity, by providing a new upper bound on its feedforward capacity that lies slightly below the feedback capacity. 

The remainder of this paper is organized as follows. Section \ref{sec:prelimi} introduces the notation and the model definition. Section \ref{sec: methodologies} introduces computable upper and lower bounds on the capacity of unifilar FSCs with instantaneous feedback. Section \ref{seq:main_results} states the main results of this paper. Section \ref{sec: delayed_feedback} presents our demonstration of the fact that the delayed feedback capacity problem can be reduced into an instantaneous feedback capacity problem, by appropriately reformulating the channel. Subsequently, we also introduce computable upper and lower bounds on the delayed feedback capacity of unifilar FSCs. Section \ref{sec: trapdoor} presents our main results regarding the capacity of the trapdoor channel. Section \ref{sec: upper_bounds} provides novel results concerning the feedforward capacities of the input-constrained BSC and the DEC, by investigating their delayed feedback capacity. Finally, our conclusions appear in Section \ref{sec:conclusion}. To preserve the flow of the presentation, some of the proofs are given in the appendices.

\section{Notation and Preliminaries}\label{sec:prelimi}
In this section, we introduce the notation, the model definition, and our MDP framework.
\subsection{Notation}
Throughout this paper, random variables will be denoted by capital letters and their realizations will be denoted by lower-case letters, e.g. $X$ and $x$, respectively. Calligraphic letters denote sets, e.g.  $\mathcal{X}$. We use the notation $X^n$ to denote the random vector $(X_1,X_2,\dots,X_n)$ and $x^n$ to denote the realization of such a random vector. We also use the notation $X_i^j$ ($j>i$) to designate the random vector $(X_i,\dots,X_j)$. For a real number $\alpha\in [0,1]$, we define $\bar{\alpha}=1-\alpha$.
The probability $\Pr[X=x]$ is denoted by $P_X(x)$. When the random variable is clear from the context, we write it in shorthand as $P(x)$.
The directed information between $X^n$ and $Y^n$ is defined as
\begin{align*}
    I(X^n\rightarrow Y^
    n) = \sum_{i=1}^n I(X^i;Y_i|Y^{i-1}).
\end{align*}
The probability mass function of $X^n$ \textit{causally conditioned} on $Y^{n-d}$ is defined as
\begin{align*}
    P(x^n\|y^{n-d}) = \prod_{i=1}^n P(x_i|x^{i-1},y^{i-d}).
\end{align*}
We denote by $\mathrm{C}$, $\mathrm{C}^\mathrm{fb}_1$, and $\mathrm{C}^\mathrm{fb}_\mathrm{d}$, the feedforward capacity (i.e. no feedback), the feedback capacity (the capacity with instantaneous feedback, i.e. $d=1$ in Fig. \ref{fig:FSC}), and the $d$ time instances delayed feedback capacity, respectively. 

\subsection{Finite-state Channels} 
A FSC is defined statistically by a time-invariant transition probability kernel, $P_{S^+,Y|X,S}$, where $X$, $Y$, $S$, $S^+$ denote the channel input, output, and state before and after one transmission, respectively. The cardinalities $\mathcal{X},\mathcal{Y},\mathcal{S}$ are assumed to be finite. Formally, given a message $m$, the channel has the following property:
\begin{align}\label{eq: channel_markov}
    P(s_t,y_t|x^t,y^{t-1},s^{t-1},m) = P_{S^+,Y|X,S}(s_t,y_t|x_t,s_{t-1}).
\end{align}
A unifilar FSC has the additional property that the state evolution is given by a time-invariant function, $f(\cdot)$, such that $s_t = f(s_{t-1},x_t,y_t)$.

As shown in the theorem below, the feedback capacity of a strongly connected\footnote{A FSC is strongly connected if for any states $s,s^{\prime}\in\mathcal{S}$, there exit an integer $T$ and an input distribution $\{P_{X_t|S_{t-1}}\}_{t=1}^{T}$ such that $\sum_{t=1}^T P_{S_t|S_0}(s|s^{\prime})>0$.} FSC is given by a multi-letter expression that cannot be computed directly. 
\begin{theorem}[\!\cite{PermuterCuffVanRoyWeissman08}, Th. 3] \label{FSC_feedback_Capacity}
 The feedback capacity of a strongly connected FSC is
\begin{align}
	\mathrm{C}^\mathrm{fb}_1 = \lim_{n\to\infty}\frac{1}{n}\max_{P(x^n\|y^{n-1})}I(X^n\rightarrow Y^n).
\end{align}
\end{theorem}

In this paper, we consider a communication setting that involves delayed feedback of $d$ time instances, as depicted in Fig. \ref{fig:FSC}. In particular, at time $t$, the encoder has access to the message $m$ and the channel outputs up to time $t-d$. As a result, the encoder output $x_t$ is a function of both the message and the delayed feedback. The channel input $x_t$ then goes through a FSC, and the resulting output $y_t$ enters the decoder. The encoder then receives the feedback sample with a delay of $d$ time instances. When the feedback has a delay of $d$ time instances, the maximization over the directed information in Theorem \ref{FSC_feedback_Capacity} is performed over $P(x^n\|y^{n-d})$ instead of over $P(x^n\|y^{n-1})$ \cite{PermuterWeissmanGoldsmith}.

\subsection{MDP Framework}
In this section, we introduce the MDP problem, which is a mathematical framework for modeling decision-making problems in which the outcomes of actions are uncertain and dependent on the current state of the system. Specifically, we consider an MDP problem with a state space $\mathcal{Z}$, an action space $\mathcal{U}$, and a disturbance space $\mathcal{W}$. The initial state $z_0\in\mathcal{Z}$ is randomly drawn from a distribution $P_Z$. At each time step $t$, the system is in a state $z_{t-1}\in\mathcal{Z}$, the decision-maker selects an action $u_t\in\mathcal{U}$, and a disturbance $w_t\in\mathcal{W}$ is drawn according to a conditional distribution $P_w(\cdot|z_{t-1},u_t)$. The state $z_t$ then evolves according to a transition function $F:\mathcal{Z}\times\mathcal{U}\times\mathcal{W}\to\mathcal{Z}$, i.e., $z_{t}=F(z_{t-1},u_{t},w_{t})$. The disturbance $w_t$ represents uncertainty and exogenous influences that affect the system's evolution over time. These disturbances capture external factors and stochastic elements beyond the control of the decision-maker.

The decision-maker selects the action $u_t$ according to a deterministic function $\mu_t$, which maps histories $h_t=(z_0,w_0,\dots,w_{t-1})$ onto actions, i.e. $u_t = \mu_t(h_t)$. Note that given history $h_t$ and a policy $\pi=\{\mu_1,\mu_2,...\}$, one can compute past states $z_1,\dots,z_{t-1}$ and actions $u_1,\dots,u_{t-1}$. Given a policy $\pi$ and a bounded reward function $g:\mathcal{Z}\times\mathcal{U}\to\mathbb{R}$, our goal is to maximize the average reward over an infinite time horizon. The average reward achieved by policy $\pi$ is defined as $$\rho_\pi = \liminf_{n\to\infty}\frac{1}{n}\mathbb{E}_\pi\left[\sum_{t=0}^{n-1}g\left(Z_t,\mu_{t+1}(h_{t+1})\right)\right].$$ The optimal average reward is denoted by $\rho^*$ and is achieved by the policy that maximizes the expected sum of rewards over time, i.e., $\rho^*=\sup_{\pi}\rho_\pi$.

The following theorem presents the Bellman equation in the context of the formulation defined above. For MDP problems, this equation provides a sufficient condition for determining whether a given average reward is optimal.
\begin{theorem}[Bellman equation, \cite{Arapos93_average_cose_survey}]\label{Theorem:Bellman_original}
    If a scalar $\rho\in\mathbb{R}$ and a bounded function  $h:\mathcal{Z}\rightarrow\mathbb{R}$ satisfy
    \begin{align*}
        \rho+h(z) = \sup_{u\in\mathcal{U}}\pr{g\pr{z,u}+\int    P_w(dw|z,u) h\pr{F\pr{z,u,w}}},\\
        \forall z\in\mathcal{Z}
    \end{align*}
    then $\rho=\rho^{*}$.
\end{theorem}

\section{Bounds on Feedback Capacity}\label{sec: methodologies}
In this section, we introduce computable bounds on the feedback capacity of unifilar FSCs. These bounds were introduced for unifilar FSCs with instantaneous feedback in \cite{Sabag_UB_IT,Sabag_DB_FB}. In Section \ref{sec: delayed_feedback}, we demonstrate that they can be extended to the case of delayed feedback as well. 
\subsection{The $Q$-graph Bounds}\label{sec: method1_qbounds}
We begin by introducing an auxiliary tool known as the \textit{$Q$-graph}. For a fixed $Q$-graph, we then present the single-letter upper and lower bounds that were established in \cite{Sabag_UB_IT}. The $Q$-graph is a directed, connected, and labeled graph, for which each of its nodes have $|\mathcal{Y}|$ outgoing edges with distinct labels from the channel output alphabet. Given an initial node, an output sequence, $y^t$, is mapped onto a unique node by walking along the labeled edges. An example of a $Q$-graph is provided in Fig. \ref{fig:1Markov}. The induced mapping is denoted by $\Phi_{t}:{\cY}^{t}\to \cQ$, which can be presented alternatively as a function $\phi:\cQ\times\cY\to\cQ$. Namely, a new graph node can be computed as a time-invariant function of the previous node and a channel output.
\begin{remark}
A special case of a $Q$-graph is a \emph{$k$th-order Markov $Q$-graph}, which is defined on the set of nodes $\cQ = \cY^k$; for each node $q = (y_1,y_2,\ldots,y_k)$, the outgoing edge labeled $y \in \cY$ goes to the node $(y_2,\ldots,y_k,y)$. For instance, Fig.~\ref{fig:1Markov} shows a Markov $Q$-graph with $\mathcal{Y} = \{ 0, 1 \}$ and $k=1$. 
\end{remark}

For a fixed FSC and a given $Q$-graph, consider the $(S,Q)$-graph, an additional directed graph that combines both the information on the $Q$-graph and the evolution of the channel states. Specifically, split each node in the $Q$-graph into $|\cS|$ new nodes, which are represented by pairs $(s,q)\in\cS\times\cQ$. Then, an edge labeled $(x,y)$ from node $(s,q)$ to node $(s^+,q^+)$ exists if and only if there is a pair $(x,y)$ such that $s^+=f(s,x,y)$, $q^+=\phi(q,y)$, and $P(y|x,s)>0$. The pair of functions $(f,\phi)$ are given by the channel and the fixed $Q$-graph. For any choice of input distribution $P_{X|S,Q}$, the transition probabilities on the edges of the $(S,Q)$-graph are computed as
\begin{align}\label{eq: sq_transition}
        &P(s^+,q^+|s,q) \nn\\
        &= \sum_{x,y}P(x,y,s^+,q^+|s,q)\nn\\
        &\stackrel{(a)}= \sum_{x,y}P(x|s,q)P(y|x,s)\mathbbm{1}_{\{q^+=\phi(q,y)\}}\mathbbm{1}_{{\{s^+=f(s,x,y)\}}},
\end{align}
where $(a)$ follows by the channel law and by the fact that $q^+$ is a deterministic function of $(q,y)$. We define the notation $\mathcal{P}_{\pi}$ as the set of input distributions $P_{X|S,Q}$ that induce a unique stationary distribution on $(S,Q)$, namely, their corresponding $(S,Q)$-graph is irreducible and aperiodic.

\begin{figure}[tb]
\centering
    \psfrag{Q}[][][0.8]{$\;Y=1$}
    \psfrag{E}[][][0.8]{$\;\;Y=0$}
    \psfrag{L}[][][0.8]{$Q=1$}
    \psfrag{H}[][][0.8]{$Q=0$}
    \includegraphics[scale = 0.45]{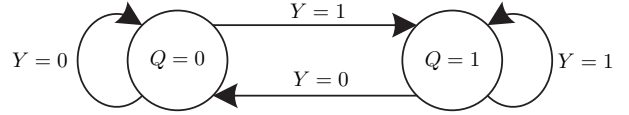}
    \caption{An example of a $Q$-graph with $\mathcal{Q}=2$ and $\mathcal{Y}=\{0,1\}$.}
    \label{fig:1Markov}
\end{figure}

In the following theorem we introduce the upper bound.
\begin{theorem}\cite[Theorem $2$]{Sabag_UB_IT}\label{theorem:Q_UB}
The feedback capacity of a strongly connected unifilar FSC, where the initial state is available to both the encoder and the decoder, is bounded by
\begin{align}\label{eq:Theorem_Upper}
\mathrm{C}^{\mathrm{fb}}_\mathrm{1}\leq \sup_{P_{X|S,Q}\in\mathcal{P}_{\pi}}I(X,S;Y|Q),
\end{align}
for all $Q$-graphs for which the $(S,Q)$-graph has a single and aperiodic closed communicating class. The joint distribution is $P_{Y,X,S,Q}=P_{Y|X,S}P_{X|S,Q}\pi_{S,Q}$, where $\pi_{S,Q}$ is the stationary distribution of the $(S,Q)$-graph.
\end{theorem}

We proceed to describe the lower bound. Let us first define a property called the \textit{BCJR-invariant input}. An input distribution $P_{X|S,Q}$ is said to be an \textit{aperiodic input}
if its $(S,Q)$-graph is aperiodic, and an aperiodic input distribution is said to be \textit{BCJR-invariant} if the Markov chain $S^+-Q^+-(Q,Y)$ holds. A simple verification of the Markov chain is given by the following equation:
\begin{align}\label{eq: BCJR}
    \pi(s^+|q^+) = \frac{\sum_{x,s}\pi(s|q)P(x|s,q)P(y|x,s)\mathbbm{1}_{\{s^+=f(x,y,s)\}}}{\sum_{x',s'}\pi(s'|q)P(x'|s',q)P(y|x',s')},
\end{align}
which needs to hold for all $(s^+, q, y)$ and $q^+ = \phi(q,y)$.

Now that we have defined a BCJR-invariant input distribution, the lower bound can be introduced.
 \begin{theorem}\cite[Theorem $3$]{Sabag_UB_IT}\label{theorem:bcjr_general}
If the initial state $s_0$ is available to both the encoder and the decoder, then the feedback capacity of a strongly connected unifilar FSC is bounded by
\begin{align}\label{eq:Theorem_Lower}
\mathrm{C}^{\mathrm{fb}}_\mathrm{1}&\geq I(X,S;Y|Q),
\end{align}
for all aperiodic inputs $P_{X|S,Q}\in\mathcal{P}_\pi$ that are BCJR-invariant, and for all irreducible $Q$-graphs with $q_0$ such that $(s_0,q_0)$ lies in an aperiodic closed communicating class.
\end{theorem}
Henceforth, we refer to a pair of a $Q$-graph and an input distribution $P_{X|S,Q}$ that satisfies the BCJR-invariant property as a \emph{graph-based encoder}.

\begin{remark}\label{remark: coding_scheme}
The upper bound in Theorem \ref{theorem:Q_UB} can be formulated as a convex optimization \cite{OronBasharfeedback}. As a result, for a fixed $Q$-graph, the upper bound can be efficiently evaluated. On the other hand, the lower bound optimization results in a non-convex optimization problem, but it still has the advantage that any feasible point (i.e. BCJR-invariant input distribution) induces a graph-based encoder. It was shown in \cite{OronBasharfeedback} that any graph-based encoder implies a simple coding scheme that achieves the lower bound. The scheme is based on the posterior matching principle in \cite{shayevitz_posterior_mathcing}, but for channels with memory whose details are given in \cite{OronBasharfeedback,Sabag_BIBO}.
\end{remark}
\begin{remark}
The selection of the $Q$-graph significantly impacts the performance of the bounds. To identify a $Q$-graph with good performance, we suggest conducting an exhaustive search over all possible $Q$-graphs, as explained in detail in \cite{OronBasharfeedback}. While such an exploration can become computationally expensive with increasing $Q$-graph size, it is often sufficient to consider a small cardinality graph to obtain good performance or capacity-achieving bounds. Another approach is to evaluate the performance of Markov $Q$-graphs, which often provide good performance bounds. A Matlab implementation of the optimization problems, including the $Q$-graph search methods, is available in \cite{Github_OronBashar}.
\end{remark}

\subsection{Upper Bounds via Duality}\label{sec: method2_DB}
Here we present computable upper bounds on the capacity of unifilar FSCs from \cite{Sabag_DB_FB}, that are based on the dual capacity upper bound \cite{Topsoe67}. For the sake of clarity, we first introduce the dual capacity upper bound for a discrete memoryless channel. Specifically, for a memoryless channel, $P_{Y|X}$, and for any choice of a test distribution, $T_Y$, on the channel output alphabet, the dual capacity upper bound states that
\begin{align}\label{DB_MC}
\mathrm{C}\leq \max_{x\in\mathcal{X}} D\pr{P_{Y|X=x} \| T_{Y}}.
\end{align}

The choice of the test distribution is crucial since it directly affects the performance of the bound. If the test distribution is equal to the optimal output distribution, then the upper bound is tight. For FSCs, the dual upper bound depends on a test distribution, $T_{Y^n}$, with memory. In \cite{Huleihel_Sabag_DB,Sabag_DB_FB}, test distributions that are structured on a $Q$-graph were proposed, that is, the following equality holds:
\begin{align}\label{eq:test_dist}
	T_{Y^n}(y^n) = \prod_{t=1}^n T_{Y|Q}(y_t|q_{t-1}),
\end{align}
where $q_{t-1}=\Phi(y^{t-1})$. We refer to such test distributions as \emph{graph-based test distributions}. The use of graph-based test distributions yielded the result in the theorem below.
\begin{theorem}[Computable upper bounds]\cite[Theorem $4$]{Sabag_DB_FB}\label{theorem: DUB_Q}
For any graph-based test distribution $T_{Y|Q}$, the feedback capacity of a strongly connected unifilar FSC is bounded by
\begin{align}\label{eq: DUB_Q}
    \mathrm{C}^{\mathrm{fb}}_\mathrm{1}& \le \lim_{n\to\infty}\max_{f(x^n\|y^{n-1})}\min_{s_0,q_0}\nn\\&\frac{1}{n}\sum_{i=1}^n\mathbb{E}\left[D\left(P_{Y|X,S}(\cdot|x_i,S_{i-1})\Bigg{\|}T_{Y|Q}(\cdot|Q_{i-1})\right)\right],
\end{align}
where $f(x^n\|y^{n-1})$ stands for causal conditioning of deterministic functions, i.e.
\begin{align}
    f(x^n\|y^{n-1}) = \prod_i \mathbbm{1}_{\{x_i=f_i(x^{i-1},y^{i-1})\}}.\nn
\end{align}
Additionally, the upper bound in \eqref{eq: DUB_Q} defines an infinite-horizon average reward MDP that is presented in Table \ref{table: main1}. 
\end{theorem}

\begin{table}[b]
\caption{MDP Formulation}
\label{table: main1}
\begin{center}
\scalebox{1.1}{
 \begin{tabular}{|c | c |} 
 \hline
MDP notations & Upper bound on capacity \\ [0.5ex] 
 \hline\hline
MDP state & $z_{t-1} \triangleq (s_{t-1}, q_{t-1})$ \\ [0.5ex]
 \hline
Action & $u_t \triangleq x_t$ \\[0.5ex]
 \hline
Disturbance & $w_t \triangleq y_t$ \\[0.5ex]
 \hline
The reward & $D\left(P_{Y|X,S}(\cdot|x_t,s_{t-1})\middle\|T_{Y|Q}(\cdot|q_{t-1})\right)$\\[0.5ex]
 \hline
\end{tabular} }
\end{center}
\end{table}

The following theorem is a simplification of the Bellman equation in Theorem \ref{Theorem:Bellman_original} for the case of the MDP formulation in Table \ref{table: main1}. 
\begin{theorem}[Bellman equation]\label{Theorem:Bellman}
If there exists a scalar $\rho\in\mathbb{R}$ and a bounded function $h:\mathcal{S}\times\mathcal{Q}\rightarrow\mathbb{R}$ such that
\begin{align}\label{eq:Bellman}
    \rho+h(s,q) &= \max_{x\in\mathcal{X}}\Big(D\left(P_{Y|X,S}(\cdot|x,s)\middle\|T_{Y|Q}(\cdot|q)\right)\nn\\&+\sum_{y\in\mathcal{Y}} P(y|x,s)h\pr{f(s,x,y),\phi(q,y)}\Big),
\end{align}
for all $(s,q)$, then $\rho=\rho^{*}$.
\end{theorem}
Following Theorem \ref{Theorem:Bellman}, it is sufficient to solve the Bellman equation associated with the MDP problem in Table \ref{table: main1} to show that the feedback capacity of a given FSC is upper bounded by the induced average reward of the MDP. Note that the MDP formulation in Table \ref{table: main1} consists of finite MDP states, actions, and disturbances. Thus, given an optimal policy, it is relatively easy to derive a conjectured solution for $(\rho^*,h(\cdot))$ since it is based on solving a finite set of linear equations. 
\begin{remark}
    Although the $Q$-graph upper bound has yielded new capacity results, its analytical computation remains challenging due to the necessity of verifying the Karush–Kuhn–Tucker (KKT) conditions, especially when dealing with channel parameters involving large alphabets. The approach proposed in this section provides an alternative and notably simpler method for deriving analytical upper bounds. However, it is essential to note that these bounds heavily depend on the selection of the test distribution, which must be chosen carefully to derive meaningful upper bounds.
\end{remark}

\section{Main Results}\label{seq:main_results}
This section summarizes the main results of this paper. We start with a general result, in which the delayed feedback capacity of any FSC can be computed as the instantaneous feedback capacity of a transformed FSC. By utilizing this reduction, we derive computable upper and lower bounds on the delayed feedback capacity of unifilar FSCs, which will be introduced in Section \ref{sec: delayed_new_formulas}. 

For a positive integer $d$ and for any FSC given by a transition kernel $P_{Y,S^+|X,S}$, we define the following transformation:
\begin{itemize}
    \item The channel state is $\tilde{S}_{t}\triangleq\left(S_{t-d+1},X_{t-d+2}^{t}\right)$.
    \item The channel output is $\tilde{Y}_t\triangleq Y_{t-d+1}$.
    \item The channel input remains the same, i.e. $\tilde{X}_t \triangleq X_t$.
\end{itemize}
First, it can be shown that the above transformation defines a new FSC with a transition kernel $P^d_{\tilde{Y},\tilde{S}^+|\tilde{X},\tilde{S}}$, where the superscript $d$ emphasizes the transformation dependence on the delay $d$. That is, the new channel follows the time-invariant Markov property of FSCs in \eqref{eq: channel_markov}. Second, we show in the following theorem the relation between the channel capacity of the original FSC and its transformation.
\begin{theorem}\label{th: delay_as_one}
The capacity of a FSC $P_{Y,S^+|X,S}$ with delayed feedback of $d$ time instances is equal to the instantaneous feedback capacity of the FSC $P^d_{\tilde{Y},\tilde{S}^+|\tilde{X},\tilde{S}}$. Furthermore, if the original FSC is unifilar, then the new transformed FSC is unifilar as well.
\end{theorem}
The proof of Theorem \ref{th: delay_as_one} is shown in Section \ref{sec: delayed_as_oneunit}. In Section \ref{sec: methodologies} we introduced two powerful methodologies to compute upper and lower bounds on the capacity of unifilar FSCs with instantaneous feedback. Based on these approaches, we establish computable upper and lower bounds on the capacity of the unifilar FSC with delayed feedback. Specifically, following Theorem \ref{th: delay_as_one}, since the delayed feedback capacity of a unifilar FSC can be computed as the capacity of a new unifilar FSC with instantaneous feedback, the computable bounds from Section \ref{sec: methodologies} can be directly adapted for the case of delayed feedback, just by redefining the channel and then applying the bounds on the new unifilar FSC. It is important to note that the cardinality of the new channel state is $|\tilde{\mathcal{S}}|=|\mathcal{S}|\cdot|\mathcal{X}|^{d-1}$, while the cardinality of the channel state in the original channel is $|\mathcal{S}|$. That is, we pay the price of having a larger channel state space as a result of the delay. When $d$ grows to infinity, the feedforwad capacity is considered. However, performing asymptotic analysis to obtain bounds on the feedforward capacity becomes intractable since the channel state cardinality of the transformed channel explodes.

\begin{figure}[t]
    \centering
    \psfrag{B}[][][1]{$1$}
    \psfrag{C}[][][1]{$0$}
    \psfrag{D}[][][1]{$1$}
    \psfrag{E}[][][1]{$0$}
    \psfrag{F}[][][1]{$1$}
    \psfrag{G}[][][1]{$1$}
    \psfrag{J}[][][1]{$x_{t+2}$}
    \psfrag{K}[][][1]{$x_{t+1}$}
    \psfrag{L}[][][1]{$x_t$}
    \psfrag{M}[][][1]{$s_{t-1}$}
    \psfrag{N}[][][1]{$y_{t-1}$}
    \psfrag{O}[][][1]{$y_{t-2}$}
    \psfrag{Q}[][][1]{Channel}
    \includegraphics[scale = 1.1]{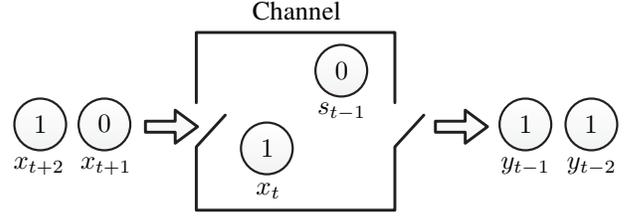}
    \caption{The trapdoor channel. The channel can be viewed as a box in which at time $t$ a labelled ball $s_{t-1}$ (channel state) lies. Then, a new ball $x_t$ (channel input) is inserted into the box, and the channel output $y_t$ is chosen with equal probability as either $s_{t-1}$ or $x_t$. The remaining ball in the box (either $s_{t-1}$ or $x_t$) is now called $s_t$ and serves as the channel state for the next time-instance.}
    \label{fig:TrapdoorChannel}
\end{figure}
In Sections \ref{sec: trapdoor} and \ref{sec: upper_bounds}, we present novel results concerning the capacity (with and without feedback) of several well-known FSCs, by investigating their delayed feedback capacity. In particular, in the following theorems, we present the main results along with some brief remarks. The first theorem below regards the capacity of the trapdoor channel, which is described in Fig. \ref{fig:TrapdoorChannel}.
\begin{theorem} \label{theorem: Trapdoor}
The capacity of the trapdoor channel with delayed feedback of two time instances is
\begin{align*} 
\mathrm{C}^{\mathrm{fb}}_\mathrm{2} = \log_2\left(\frac{3}{2}\right). 
\end{align*}
\end{theorem}
The proof of Theorem \ref{theorem: Trapdoor} is given in Appendix \ref{app: trapdoor}. A detailed discussion regarding the result is given in Section \ref{sec: trapdoor}. Its main implications are:
\begin{enumerate}
    \item There is a simple coding scheme that is based on the posterior matching scheme \cite{OronBasharfeedback}, which achieves the capacity result in Theorem \ref{theorem: Trapdoor}.
    \item We show that our capacity result is not equal to the feedforward capacity. That is, computer-based simulations with a greater delay of the feedback suggest that the feedforward capacity satisfies $C\le 0.5765$, while it is known to be lower bounded by $0.572$ \cite{Kobayashi02}. 
\end{enumerate}

In the theorem below we present the second result concerning the capacity of the BSC with a no consecutive ones input constraint.
\begin{theorem}\label{theorem: BSC_UB}
The capacity of the input-constrained BSC($p$) with a $(1,\infty)$-RLL constraint and delayed feedback of two time instances satisfies 
\begin{align*} 
   \mathrm{C}^{\mathrm{fb}}_\mathrm{2}(p)\leq  \min\log_2\left(\frac{p^p\bar{p}^{\bar{p}}a^{(p^3-3p^2+3p-1)}(\bar{b}\bar{c}d)^{(p^3-p^2)}}{(\bar{a}bc)^{(p^3-2p^2+p)}\bar{d}^{p^3}}\right), 
\end{align*}
where the minimum is over all $(a,b,c,d)\in[0,1]^4$ that satisfy:
\begin{align}\label{eq: bsc_cons}
    1 &\leq \frac{a^{(4p^3-12p^2+11p-3)}(\bar{b}d)^{(4p^3-6p^2+2p)}\bar{c}^{(4p^3-4p^2+p)}}{(\bar{a}c)^{(4p^3-8p^2+5p-1)}b^{(4p^3-10p^2+6p-1)}\bar{d}^{(4p^3-2p^2)}}, \nn\\
    1 &\leq \frac{a^{(4p^3-10p^2+8p-2)}(\bar{b}d)^{(4p^3-4p^2+p)}\bar{c}^{(4p^3-2p^2-2p+1)}}{(\bar{a}c)^{(4p^3-6p^2+2p)}b^{(4p^3-8p^2+5p-1)}\bar{d}^{(4p^3-p)}}. 
\end{align}
\end{theorem}
As shown in the theorem above, we derived an analytical upper bound on the delayed feedback capacity of the BSC with a no consecutive ones input constraint. Although we introduce our bound for the case of feedback delayed by two time instances, it also provides a novel upper bound on the feedforward capacity that outperforms all previously known bounds. Nonetheless, we emphasize that our bound almost coincides with a lower bound on the feedforward capacity. 
The proof of Theorem \ref{theorem: BSC_UB} is given in Appendix \ref{app: BSC}, and a detailed discussion regarding the result is given in Section \ref{sec: BSC}.

\begin{figure}[t]
    \centering
    \psfrag{A}[][][1]{$0$}
    \psfrag{B}[][][1]{$1$}
    \psfrag{C}[][][1]{$?$}
    \psfrag{D}[][][1]{$-1$}
    \psfrag{E}[][][1]{$1-p$}
    \psfrag{F}[][][1]{$p$}
    \psfrag{G}[][][1]{$X_t$}
    \psfrag{H}[][][1]{$Y_t$}
    \psfrag{J}[][][1]{$X_{t-1} = 0$}
    \psfrag{K}[][][1]{$X_{t-1} = 1$}
    \includegraphics[scale = 0.65]{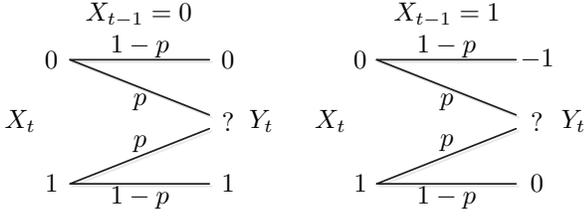}
    \caption{The DEC. The inputs take values from the binary alphabet while the outputs take values in $\mathcal{Y}=\{-1,0,1,?\}$. Given an input $x_t$, the output of the DEC is $y_t=x_t-x_{t-1}$ with probability $1-p$, or $y_t=?$ with probability $p$, where $p \in [0,1]$ is the channel parameter. The channel state is the previous input, i.e. $s_{t-1}=x_{t-1}$.}
    \label{fig: DEC_operation}
\end{figure}
The following theorem presents our result on the capacity of the DEC. The operation of the DEC is described in Fig. \ref{fig: DEC_operation}.
\begin{theorem}\label{theorem: DEC}
Feedback increases the capacity of the DEC.
\end{theorem}
In the theorem above we state that the feedback capacity of the DEC, which was derived in \cite{Sabag_UB_IT}, is not equal to its feedforward capacity. In \cite{Huleihel_Sabag_DB}, the authors investigated the feedforward capacity of the DEC, and they derived an upper bound that was equal exactly to the feedback capacity, which is known to be almost tight to a lower bound on the feedforward capacity. This fact raised the question of whether the feedback capacity is equal to the feedforward capacity. By investigating the delayed feedback capacity of the DEC, we showed that this is not the case. Further details regarding the statement of Theorem \ref{theorem: DEC} are given in Section \ref{sec:DEC}.

\begin{remark}
As will be demonstrated in Sections \ref{sec: trapdoor} and \ref{sec: upper_bounds}, the role of \emph{instantaneous feedback} is critical in terms of channel capacity. In other words, when considering both the trapdoor channel and the input constrained BSC with feedback, even a single time-instance delay leads to a sharp decrease of the capacity towards the feedforward capacity. 
\end{remark}

\section{Computable Bounds On Delayed Feedback Capacity}\label{sec: delayed_feedback}
In this section, we first show that the delayed feedback capacity problem can be converted into an instantaneous feedback capacity problem, by reformulating the FSC. Consequently, for any unifilar FSC, we then introduce upper and lower bounds on the delayed feedback capacity that are a straightforward extension of the $Q$-graph bounds from Section \ref{sec: method1_qbounds}.
\subsection{Proof of Theorem \ref{th: delay_as_one}: Delayed Feedback Capacity as Instantaneous Feedback Capacity}\label{sec: delayed_as_oneunit}
Here, we show that the capacity of a general FSC with delayed feedback can be computed as the capacity of a new FSC with instantaneous feedback, by redefining the channel state, input, and output. Moreover, if the original channel is a unifilar FSC, then the new reformulated channel is a unifilar FSC as well. 

Formally, given a general FSC, define the new channel state as $\tilde{S}_{t-1}=\left(S_{t-d},X_{t-d+1}^{t-1}\right)$, let the channel output be $\tilde{Y}_t=Y_{t-d+1}$, and let the channel input remain the same, i.e. $\tilde{X}_t=X_t$. In the following, we show that the new channel is a FSC such that its capacity with instantaneous feedback is equal to the capacity of the original channel with delayed feedback of $d$ time instances. 
\begin{itemize}
    \item Conditioned on the previous channel state $\tilde{S}_{t-1}$, the channel input $\tilde{X}_t$, and the channel output $\tilde{Y}_t$, the new channel state $\tilde{S}_t$ is independent of any other previous states, inputs, and outputs. That is,
    \begin{align}\label{eq: new_state_law}
        P(\tilde{s}_t|\tilde{x}^t,\tilde{y}^t,\tilde{s}^{t-1}) = P(\tilde{s}_t|\tilde{s}_{t-1},\tilde{x}_t,\tilde{y}_t).
    \end{align}
    This equation holds due to the Markov chain $(S_{t-d+1},X_{t-d+2}^t)-(S_{t-d},X_{t-d+1}^t,Y_{t-d+1})-(X^{t-d},Y^{t-d})$, which follows directly by the Markov chain property of the original channel. In particular, since $\tilde{S}_t=(S_{t-d+1},X_{t-d+2}^t)$, and since $(\tilde{S}_{t-1},\tilde{X}_t,\tilde{Y}_t)$ include $(S_{t-d},X_{t-d+1}^t,Y_{t-d+1})$, \eqref{eq: new_state_law} holds.
    
    In addition, we show below that, if the original channel is a unifilar FSC, then the unfilar property also holds for the new induced channel. That is, we show that the new channel state $\tilde{s}_t$ is a time-invariant function of $\tilde{s}_{t-1}$, $\tilde{x}_t$, and $\tilde{y}_t$:
    \begin{align*}
        \tilde{s}_t &= \left(s_{t-d+1},x_{t-d+2}^{t}\right)\nn\\
                    &= \left(f\left(s_{t-d},x_{t-d+1},y_{t-d+1}\right),x_{t-d+2}^{t}\right)\nn\\
                    &\triangleq \tilde{f}\left(\tilde{s}_{t-1}, \tilde{x}_t, \tilde{y}_t\right),
    \end{align*}
    where, clearly, $\tilde{f}:\tilde{\mathcal X}\times \tilde{\mathcal{Y}}\times\tilde{\mathcal{S}}\to \tilde{\mathcal S}$ is a time-invariant function of $(\tilde{s}_{t-1}, \tilde{x}_t, \tilde{y}_t)$. Thus, the unifilar property holds in this case.
    \item Conditioned on the previous channel state $\tilde{S}_{t-1}$ and the channel input $\tilde{X}_t$, the channel output $\tilde{Y}_t$ is independent of any other previous states, inputs, and outputs. Specifically, note that $\tilde{s}_{t-1}$ includes the pair $(s_{t-d},x_{t-d+1})$, and therefore it implies that
    \begin{align*}
        P(\tilde{y}_t|\tilde{x}^t,\tilde{y}^{t-1},\tilde{s}^{t-1}) = P(\tilde{y}_t|\tilde{x}_t,\tilde{s}_{t-1}),
    \end{align*}
    due to the fact that the redefined channel output $\tilde{y}_t$ is the original channel output at time $(t-d+1)$.
    \item The initial state $\tilde{S}_0$ is known both to the encoder and to the decoder, as required. 
    As shown above, the redefined channel is a FSC. Additionally, at each time-step $t$, the encoder knows all previous channel outputs $\tilde{y}^{t-1}$, as required in the case of instantaneous feedback.
\end{itemize}
Now, given an input sequence, the corresponding outputs of the redefined FSC are drawn according to the statistics of the original channel model. Furthermore, maximizing over $P(\tilde{x}^n\|\tilde{y}^{n-1})$ is equivalent to maximizing over $P(x^n\|y^{n-d})$. Therefore, we can deduce that the capacity of the original channel with delayed feedback can be computed as
\begin{align}
       \mathrm{C}^{\mathrm{fb}}_\mathrm{d} = \lim_{n\to\infty} \max_{P(\tilde{x}^n\|\tilde{y}^{n-1})}\frac{1}{n}I(\tilde{X}^n\rightarrow\tilde{Y}^n),
\end{align}
while reformulating the channel model as described above.

A similar formulation appeared in \cite{Yang05, Sabag_Kostina_Gaussian}, but only for the case where the channel state $s_t$ is a deterministic function of $s_{t-1}$ and $x_t$. Here, we presented a general formulation that holds for any FSC. The trapdoor channel, for instance, does not fall into the framework of \cite{Yang05} and \cite{Sabag_Kostina_Gaussian} since the channel state depends on the channel outputs as well.
\begin{remark}
Following our formulation, it is interesting to observe that the channel output $\tilde{Y}_t$ is independent of $\tilde{X}_t$ conditioned on $\tilde{S}_{t-1}$. In other words, the channel output solely depends on the channel state and not on the channel input. However, the choice of the channel input $\tilde{x}_t$ is still of significant importance since it directly affects the evolution of the next channel state.
\end{remark}

\subsection{$Q$-graph Bounds on Delayed Feedback Capacity}\label{sec: delayed_new_formulas}
In this section, we present upper and lower bounds on the delayed feedback capacity of unifilar FSCs. These bounds are based on the $Q$-graph bounds from Section \ref{sec: method1_qbounds}, which were introduced for the case of instantaneous feedback. In particular, to compute bounds on the delayed feedback capacity of a given unifilar FSC, we first reformulate the FSC (according to the formulation in Section \ref{sec: delayed_as_oneunit}) to obtain an equivalent instantaneous feedback capacity problem of a new uniflar FSC. Then, given the new unifilar FSC, the $Q$-graph bounds can be directly applied. We emphasize that a delay of at least two time instances is assumed here. Otherwise, we have the standard instantaneous feedback scenario. In the following theorem, we present the $Q$-graph upper bound for the case of delayed feedback.
\begin{theorem}\label{theorem:Q_UB_delay}
The $d$ time instances delayed feedback capacity of a strongly connected unifilar FSC, where the initial state is available to both the encoder and the decoder, is bounded by
\begin{align}\label{eq:Theorem_Upper}
\mathrm{C}^{\mathrm{fb}}_\mathrm{d}\leq \sup_{P_{\tilde{X}|\tilde{S},Q}\in\mathcal{P}_{\pi}}I(\tilde{S};\tilde{Y}|Q),
\end{align}
where $\tilde{X}$, $\tilde{Y}$, $\tilde{S}$ are the new channel input, output, and state, respectively (as defined in Section \ref{sec: delayed_as_oneunit}). The bound holds for all $Q$-graphs for which the $(\tilde{S},Q)$-graph has a single and aperiodic closed communicating class. The joint distribution is $P_{\tilde{Y},\tilde{X},\tilde{S},Q}=P_{\tilde{Y}|\tilde{X},\tilde{S}}P_{\tilde{X}|\tilde{S},Q}\pi_{\tilde{S},Q}$, where $\pi_{\tilde{S},Q}$ is the stationary distribution of the $(\tilde{S},Q)$-graph.
\end{theorem}

Further, the $Q$-graph lower bound for the case of delayed feedback is given in the theorem below.
\begin{theorem}\label{theorem:bcjr_general_delay}
If the initial state $\tilde{s}_0$ is available to both the encoder and the decoder, then the $d$ time instances delayed feedback capacity of a strongly connected unifilar FSC is bounded by
\begin{align}\label{eq:Theorem_Lower}
\mathrm{C}^{\mathrm{fb}}_\mathrm{d}&\geq I(\tilde{S};\tilde{Y}|Q),
\end{align}
where $\tilde{X}$, $\tilde{Y}$, $\tilde{S}$ are the new channel input, output, and state, respectively (as defined in Section \ref{sec: delayed_as_oneunit}). The bound holds only for aperiodic inputs $P_{\tilde{X}|\tilde{S},Q}\in\mathcal{P}_\pi$ that are BCJR-invariant, and for all irreducible $Q$-graphs with $q_0$ such that $(\tilde{s}_0,q_0)$ lies in an aperiodic closed communicating class.
\end{theorem}
\begin{proof}[Proof of Theorems \ref{theorem:Q_UB_delay} and \ref{theorem:bcjr_general_delay}]
First, as shown in Section \ref{sec: delayed_as_oneunit}, after reformulating the FSC, we obtain an equivalent instantaneous feedback capacity problem. Since the new induced channel is a unifilar FSC, Theorems \ref{theorem:Q_UB} and \ref{theorem:bcjr_general} can be directly applied on the new unifilar channel. Accordingly, it is only left to show that it is sufficient to optimize $I(\tilde{S};\tilde{Y}|Q)$ instead of $I(\tilde{X},\tilde{S};\tilde{Y}|Q)$. The latter holds by the trivial Markov chain $\tilde{Y}-\tilde{S}-\tilde{X}$. In particular, according to the new formulation, the new channel state already include the new channel input.
\end{proof}
\section{Trapdoor Channel with Delayed Feedback}\label{sec: trapdoor}
The trapdoor channel (Fig. \ref{fig:TrapdoorChannel}) has had a long history in information theory since its introduction by David Blackwell in $1961$ \cite{Blackwell_trapdoor}. The channel has attracted much interest since its representation is very simple, yet its capacity computation is highly non-trivial. The channel can be viewed as a (causal) permutation channel since the weight of the input sequence is equal to the weight of the output sequence. This channel is also termed the \emph{chemical channel}, which alludes to a physical system in which chemical concentrations are used to communicate \cite{Berger02B}. A detailed discussion on the trapdoor channel can be found in Robert Ash's book \cite{Ash65} (which even uses the channel for his book cover). 

The trapdoor channel is a unifilar FSC whose operation can be described as follows. At time $t$, let $x_t\in\{0,1\}$ be the channel input and $s_{t-1}\in\{0,1\}$ be the previous channel state. The channel input, $x_t$, is transmitted through the channel. The channel output, $y_t$, is equal to the previous state $s_{t-1}$ or to the input $x_t$, with the same probability. The new channel state is evaluated according to $s_t=x_t\oplus y_t\oplus s_{t-1}$, where $\oplus$ denotes the XOR operation.

Despite the simplicity and the extensive research efforts dedicated to trapdoor channel, e.g. \cite{Ahl_kaspi87,Ahlswede98,PermuterCuffVanRoyWeissman08,kobayashi20033,Trapdoor_Lutz,Huleihel_Sabag_DB}, 
its capacity has remained an open problem for over sixty years. Notwithstanding, the capacity is known in two important variations of the original capacity problem.
In \cite{Ahl_kaspi87,Ahlswede98}, it was shown that the zero-error capacity of the trapdoor channel is $C_0 = 0.5$ bits per channel use. This provides a lower bound which is known to be non-tight (e.g. \cite{kobayashi20033}). The other variation is the feedback capacity, which is equal to $\mathrm{C}^{\mathrm{fb}}_\mathrm{1} = \log_2\left(\frac{1+\sqrt{5}}{2}\right)\approx 0.6942$, as shown in \cite{PermuterCuffVanRoyWeissman08}. It is also known that feedback \emph{does increase} the capacity for the trapdoor channel (e.g. \cite{Trapdoor_Lutz,Huleihel_Sabag_DB}). 

In Theorem \ref{theorem: Trapdoor}, we introduced our main result concerning the delayed feedback capacity of the trapdoor channel. Specifically, we stated that the capacity of the trapdoor channel with delayed feedback of two time instances is $\mathrm{C}^{\mathrm{fb}}_\mathrm{2} = \log_2\left(\frac{3}{2}\right)$. In the following, we list several implications of this capacity result. The above capacity is approximately $\mathrm{C}^{\mathrm{fb}}_\mathrm{2} \approx 0.5849$, while the instantaneous feedback capacity is approximately $\mathrm{C}^{\mathrm{fb}}_\mathrm{1} \approx 0.6942$. 
The best lower bound to date on the feedforward capacity is $\mathrm{C}\ge0.572$ \cite{kobayashi20033}. It is interesting to note that even a single time-instance delay leads to a sharp decrease in the capacity towards the feedforward capacity.

The delayed feedback capacity in Theorem \ref{theorem: Trapdoor} also serves as an upper bound on the feedforward capacity. Overall, the best bounds on the feedforward capacity are given by
\begin{align*}
    0.572 \leq C \leq 0.5849.
\end{align*}
While the delayed feedback capacity is equal to the best upper bound on the feedforward capacity, it does not establish a new upper bound. In particular, a recent paper proposed using duality-based upper bounds on the feedforward capacity and established the same bound \cite{Huleihel_Sabag_DB}. However, their bound is for the feedforward capacity only, and therefore we still need to show a converse proof for Theorem \ref{theorem: Trapdoor}.

An interesting question is whether the delayed feedback capacity is, indeed, the feedforward capacity. Simulations of the delayed feedback capacity with a delay greater than two time instances suggest that this is not the case. In particular, by operational considerations, we have the following chain of inequalities:
\begin{align}
    C\le \ldots\le \mathrm{C}^{\mathrm{fb}}_\mathrm{3}\le\mathrm{C}^{\mathrm{fb}}_\mathrm{2}\le\mathrm{C}^{\mathrm{fb}}_\mathrm{1}.
\end{align}
As clarified, the upper bound in Theorem \ref{theorem:Q_UB_delay} can be formulated as a convex optimization problem, and its evaluation for a greater delay of the feedback gives
\begin{align}\label{eq:trapdoor_ubs}
    \mathrm{C}^{\mathrm{fb}}_\mathrm{3} \leq 0.5782, \;\;\;
    \mathrm{C}^{\mathrm{fb}}_\mathrm{4} \leq 0.5765.
\end{align}
Accordingly, these simulations suggest that the feedforward capacity satisfies  $C\le 0.5765$. In other words, the delayed feedback capacity in Theorem \ref{theorem: Trapdoor} does not seem to be the feedforward capacity, which remains an open problem.

\begin{remark}
The achievability proof of Theorem \ref{theorem: Trapdoor} is based on the $Q$-graph lower bound, which was presented in Section \ref{sec: delayed_new_formulas}. That is, the lower bound was established by showing that a particular graph-based encoder, given by the $Q$-graph in \eqref{eq:LB_q_vec} and the input distribution in \eqref{eq: trapdoor_input_dist}, provides an achievable rate of $\log_2(3/2)$. This graph-based encoder implies a simple coding scheme that, in our case, achieves the capacity. As explained in Remark \ref{remark: coding_scheme}, the scheme is based on the posterior matching principle, and the exact details regarding the constriction of the coding scheme are given in \cite{OronBasharfeedback}.
\end{remark}
\begin{remark}
Following the formulation in Section \ref{sec: delayed_as_oneunit}, we present in Fig. \ref{fig:TrapdoorChannel_delay} the trapdoor channel with delayed feedback of two time instances as a new unifilar FSC with instantaneous feedback. For the new FSC, it is interesting to note that the capacity of each individual channel (per state) is zero. Nevertheless, as we already demonstrated, the capacity of the overall FSC is not zero. This follows due to the fact that, at each time $t$, the output depends only on the previous channel state, and the choice of the current input will only participate in the evolution of the next channel state.
\end{remark}

\section{Upper Bounds on Feedforward Capacity}\label{sec: upper_bounds}
In this section, we demonstrate that the investigation of the delayed feedback capacity plays an important role in deriving upper bounds on the feedforward capacity. Specifically, besides the trapdoor channel, we present here two additional FSCs for which we derive novel results concerning their feedforward capacity by investigating their delayed feedback capacity.

\subsection{Input-Constrained BSC}\label{sec: BSC}
Regardless of whether feedback is allowed or not, memoryless channels have the same simple single-letter capacity expression \cite{shannon56}. When the inputs are constrained, however, the capacity problem is very challenging. The feedforward capacity in the presence of constrained inputs has been extensively investigated, e.g. \cite{vontobel_generalization,han_constrained_BSC_BEC,wolf_RLL,han_RLL_BSC,yonglong_isit_erasure}, but is still given by a multi-letter expression. 

\begin{figure}[t]
    \centering
    \psfrag{A}[][][0.65]{$0$}
    \psfrag{B}[][][0.65]{$1$}
    \psfrag{C}[][][0.65]{$\;0.5$}
    \psfrag{D}[][][0.7]{$\tilde{s}_{t-1}=(0,0)$}
    \psfrag{E}[][][0.7]{$\tilde{s}_{t-1}=(0,1)$}
    \psfrag{F}[][][0.7]{$\tilde{s}_{t-1}=(1,0)$}
    \psfrag{G}[][][0.7]{$\tilde{s}_{t-1}=(1,1)$}
    \psfrag{J}[][][0.7]{$\tilde{x}_t$}
    \psfrag{K}[][][0.7]{$\tilde{y}_t$}
    \includegraphics[scale = 0.47]{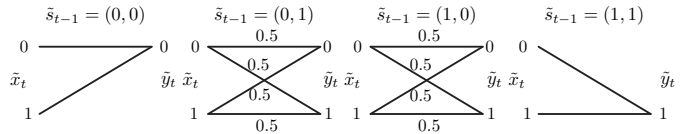}
    \caption{The trapdoor channel with delayed feedback of two time instances as a new unifilar FSC with instantaneous feedback.}
    \label{fig:TrapdoorChannel_delay}
\end{figure}
Here, we consider the BSC with crossover probability $p$, denoted by BSC($p$), where the inputs are constrained to satisfy the $(1,\infty)$-RLL constraint. Namely, the input sequence does not contain two consecutive ones. Even though this setting does not fall under the classical definition of a unifilar FSC, it is straightforward to convert input constraints by defining a dummy sink state whose capacity is zero in the case that the constraint is violated. 
For this setting, while the feedforwad capacity is still open, the feedback capacity was established in \cite{Sabag_BIBO}, and it is known that feedback does increase the capacity \cite{Sabag_BIBO,Dual_Andrew_J}. 

\begin{figure}[t]
    \centering
    \includegraphics[scale = 0.42]{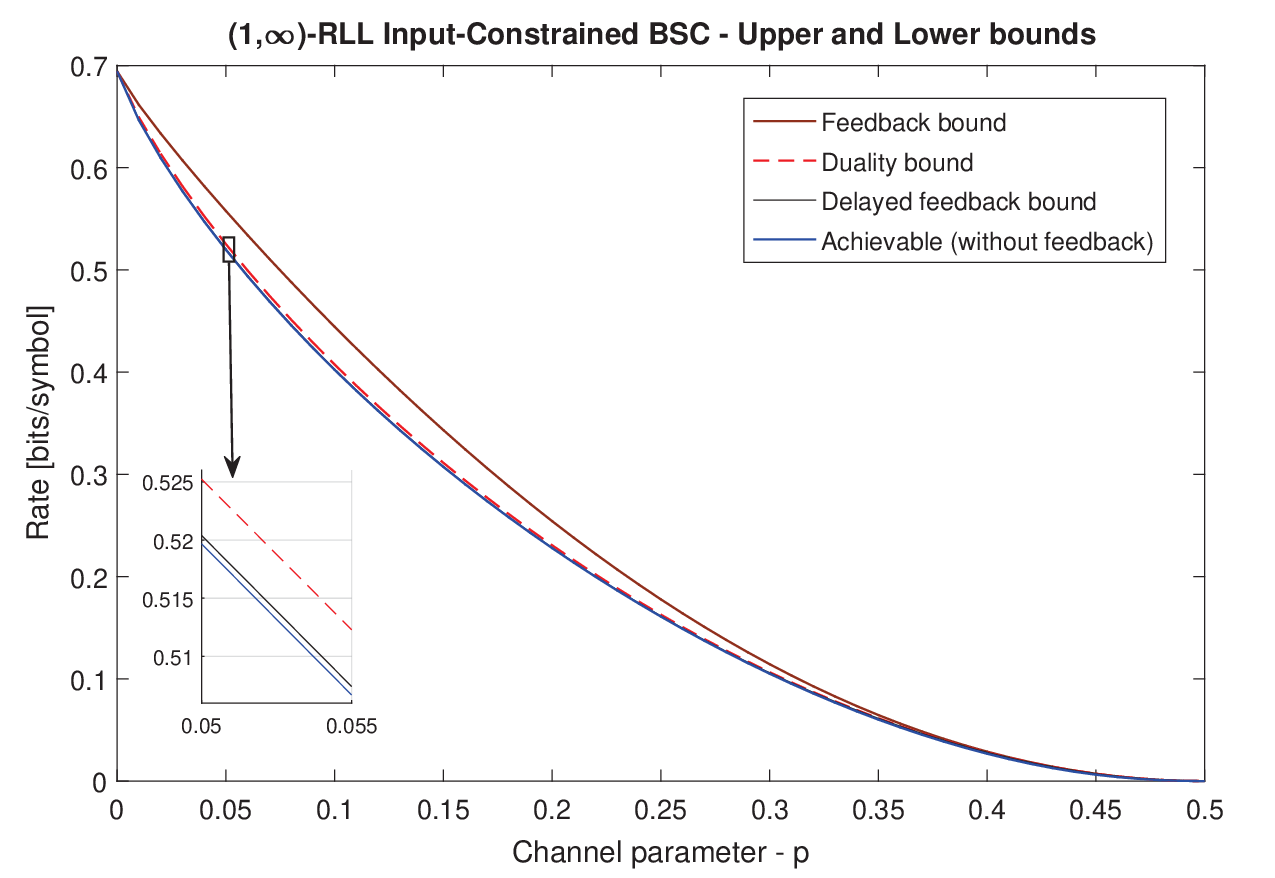}
    \caption{Upper and lower bounds on the feedforward capacity of the input-constrained BSC($p$). The feedback bound is the feedback capacity from \cite{Sabag_BIBO}. The duality bound (dashed red line) is the best upper bound from the literature \cite{Dual_Andrew_J}. The black line represents our upper bound in Theorem \ref{theorem: BSC_UB}. Finally, the bounds are compared to a lower bound on the feedforward capacity (blue line).}
    \label{fig:BSC_bounds}
\end{figure}

In Theorem \ref{theorem: BSC_UB}, we introduced our upper bound on the two time instances delayed feedback capacity of the BSC($p$) with $(1,\infty)$-RLL constraint. In Fig. \ref{fig:BSC_bounds}, this upper bound is plotted along with the feedback capacity from \cite{Sabag_BIBO}, the best upper bound on the feedforward capacity from \cite{Dual_Andrew_J}, and a lower bound on the feedforward capacity obtained using the simulation method in \cite{Loeliger_FSCs}. Here as well, it is surprising to note from the plot that the difference between the capacity with instantaneous feedback and the capacity with an additional time-instance delay is quite significant. Further, although our upper bound was introduced for the case of delayed feedback, it also serves as a novel upper bound on the feedforward capacity, and outperforms all previously known bounds. Nonetheless, our upper bound almost coincides with the achievable rate (a lower bound on the feedforward capacity).

We would like to emphasize that our upper bound can be further improved. Specifically, evaluating the upper bound in Theorem \ref{theorem:Q_UB_delay} with a $3$-order Markov $Q$-graph provides an even tighter upper bound on the delayed feedback capacity of two time instances. However, the upper bound in Theorem \ref{theorem: BSC_UB} already achieves remarkable performance, and therefore we do not provide here an analytical expression for the additional bound.

\subsection{The Dicode Erasure Channel} \label{sec:DEC}
The DEC was studied in \cite{PfitserLDPC_memory_erasure, henry_dissertation, Sabag_UB_IT,Huleihel_Sabag_DB} and is a simplified version of the well-known dicode channel with additive white Gaussian noise. The operation of the DEC is illustrated in Fig. \ref{fig: DEC_operation}. The feedback capacity of the DEC was established in \cite{Sabag_UB_IT}, and is given in the theorem below. However, in the absence of feedback the capacity is still unknown.
\begin{theorem}[\!\cite{Sabag_UB_IT}, Th. 5] \label{theorem: DEC_fb}
The feedback capacity of the DEC is
\begin{align} 
    \mathrm{C}^{\mathrm{fb}}_\mathrm{1}(p) = \max_{\epsilon\in[0,1]} (1-p)\frac{\epsilon+p H_2(\epsilon)}{p+(1-p)\epsilon}
\end{align}
for any channel parameter $p\in[0,1]$.
\end{theorem}

In the following theorem, we derive an upper bound on the delayed feedback capacity of the DEC for $p=0.5$. This bound serves as a novel upper bound on the feedforward capacity, and it also demonstrates that feedback does increase the capacity of the DEC for $p=0.5$ (as stated earlier in Theorem \ref{theorem: DEC}).
\begin{theorem}\label{theorem: DEC_UB}
The capacity of the DEC with delayed feedback of two time instances is upper bounded by
\begin{align*}
    &\mathrm{C}^{\mathrm{fb}}_\mathrm{2}(0.5) \leq\max_{a\in(0,0.5)}\nn\\& \frac{1}{4}\log_2\left(\frac{2-3a}{(1-2a)\cdot\left(1+8a^2\bar{a}-3a-(1-4a\bar{a})\sqrt{1+4a^3}\right)}\right)\hspace{-0.1cm}.
\end{align*}
\end{theorem}
The proof of Theorem \ref{theorem: DEC_UB} is given in Appendix \ref{app: DEC}. The upper bound is derived by using a particular $Q$-graph with eight nodes, which is given within the proof of the theorem.

\begin{figure}[t]
    \centering
    \includegraphics[scale = 0.43]{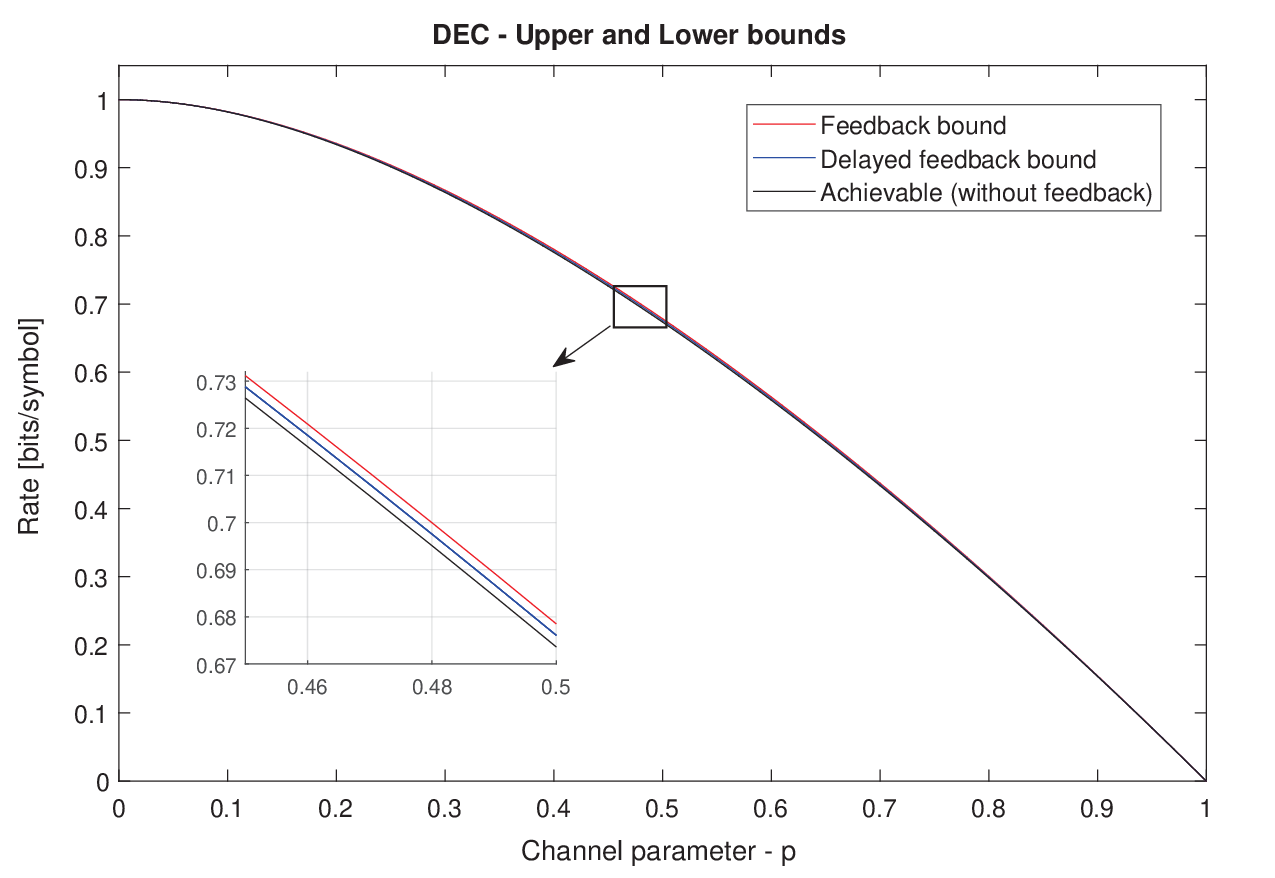}
    \caption{Upper and lower bound on the capacity of the DEC. The feedback bound is the feedback capacity from \cite{Sabag_UB_IT}. The delayed feedback bound (blue line) is our upper bound. The black line is an achievable rate on the feedforward capacity.}
    \label{fig:DEC_bounds}
\end{figure}
In Fig. \ref{fig:DEC_bounds}, we present bounds on the feedforward capacity of the DEC. In particular, the red line is the feedback capacity from Theorem \ref{theorem: DEC_fb}, which serves as a non-trivial upper bound. The black line is an achievable rate from \cite{henry_thesis_ch4}, obtained by considering a first-order Markov input process. Finally, the blue line shows our upper bound on the two time instances delayed feedback capacity. This bound is a numerical evaluation of Theorem \ref{theorem:Q_UB_delay} with the same $Q$-graph that is used for the proof of Theorem \ref{theorem: DEC_UB}. Accordingly, for $p=0.5$, the plot provides a numerical evaluation of the analytical upper bound in Theorem \ref{theorem: DEC_UB}. In \cite{Huleihel_Sabag_DB}, the authors derived an upper bound on the feedforward capacity, which turned out to be exactly equal to the feedback capacity. This fact led to the question of whether the feedback capacity is equal to the feedforward capacity, which, indeed, is negligibly different from the first-order Markov achievable rate. Following our numerical upper bound, we could see that this bound improves the feedback capacity for any $p\in(0,1)$, which indicates that feedback increases the capacity of the DEC over the entire region of the erasure parameter.

\section{Conclusions}\label{sec:conclusion}
In this paper, we investigated the delayed feedback capacity of FSCs. It was shown that the capacity of a FSC with delayed feedback can be computed as that of a new FSC with instantaneous feedback. Accordingly, several graph-based methods to obtain computable bounds on the capacity of unifilar FSCs, which were introduced for the case of instantaneous feedback, could be adapted to the case of delayed feedback as well. Using these bounds, we could establish that the capacity of the trapdoor channel with delayed feedback of two time instances is equal to $\log_2\left(\frac{3}{2}\right)$. In addition, we derived an upper bound on the delayed feedback capacity of the input-constrained BSC, which also serves as a novel upper bound on its feedforward capacity. Finally, we demonstrate that feedback increases capacity for the DEC.

\begin{appendices}
\section{Trapdoor Channel --- Proof of Capacity (Theorem \ref{theorem: Trapdoor})}\label{app: trapdoor}
\begin{proof}
The proof of the capacity result in Theorem \ref{theorem: Trapdoor} consists of two parts. In Section \ref{Trapdoor_UB}, we prove the converse, that is, $\mathrm{C}^{\mathrm{fb}}_\mathrm{2}\leq  \log_2\left(\frac{3}{2}\right)$, and in Section \ref{Trapdoor_LB}, we show a corresponding lower bound.

The proof is based on the methodology in Section \ref{sec: methodologies}. As clarified, the upper and lower bounds hold for instantaneous feedback, but using the formulation in Section \ref{sec: delayed_as_oneunit}, we are able to transform the delayed feedback capacity into a capacity problem with instantaneous feedback. 

We begin with presenting the formulation of the trapdoor channel with delayed feedback as a new unifilar FSC with instantaneous feedback. The new FSC (see Fig. \ref{fig:TrapdoorChannel_delay}) is defined as follows: the channel state consists of the pair of the previous channel state and channel input, that is, $\tilde{s}_{t-1}\triangleq(s_{t-2},x_{t-1})$. The channel input is $\tilde{x}_t=x_t$ and the channel output is $\tilde{y}_t=y_{t-1}$. If $\tilde{s}_{t-1}=(0,1)$ or $\tilde{s}_{t-1}=(1,0)$, then we have a BSC($0.5$). Otherwise, for any $\tilde{x}_t$, if $\tilde{s}_{t-1}=(0,0)$ then $\tilde{y}_t=0$, and if $\tilde{s}_{t-1}=(1,1)$ then $\tilde{y}_t=1$.

\subsection{Upper Bound}\label{Trapdoor_UB}
Here, we will show that $\mathrm{C}^{\mathrm{fb}}_\mathrm{2}\leq \log_2\left(\frac{3}{2}\right)$. The proof is based on fixing a particular graph-based test distribution, and then solving the MDP problem of the dual capacity upper bound (for additional details, see Section  \ref{sec: method2_DB}). The MDP formulation is presented in Table \ref{table: main1}.

Consider the $Q$-graph in Fig. \ref{fig:1Markov}, which consists of two nodes, and the following graph-based test distribution:
\begin{align}
    T(\tilde{y}=0|\underline{q})=\left[\frac{2}{3},\frac{1}{3}\right].
\end{align}
To present the solution for the Bellman equation, define the constant
\begin{align}\label{eq: rho_trapdoor}
    \rho^*&=\log_2\left(\frac{3}{2}\right),
\end{align}
and the value function
\begin{align}\label{eq: value_trapdoor}
    h(\tilde{s},q)=\begin{cases} 1, & (\tilde{s}=(0,0),q=2)\textit{ or }(\tilde{s}=(1,1),q=1), \\
                  0, & \textit{else}.
    \end{cases}
\end{align}
\begin{remark}
The conjectured solution $(\rho^*,h(\cdot))$ has been obtained by using the value iteration algorithm with the MDP defined in Table \ref{table: main1}. Specifically, applying the value iteration algorithm provides the optimal policy for any possible MDP state. Then, it is only left to solve a finite set of linear equations in order to derive closed-form expressions for $\rho^*$ and $h(\cdot)$.  
\end{remark}
We proceed to show that $\rho^*$ in \eqref{eq: rho_trapdoor} and the value function in \eqref{eq: value_trapdoor} solve the Bellman equation. This directly implies, by Theorem \ref{theorem: DUB_Q}, that $\mathrm{C}^{\mathrm{fb}}_\mathrm{2}\leq \rho^* = \log_2(3/2)$. 
For the MDP state $(\tilde{s}=(0,0),q=1)$, the right-hand side of the Bellman equation is a maximum over $\tilde{x}$ of
\begin{align}
    &D\left(P_{\tilde{Y}|\tilde{X},\tilde{S}}(\cdot|\tilde{x},\tilde{s})\middle\|T_{\tilde{Y}|Q}(\cdot|q)\right)\nn\\&+\sum_{\tilde{y}\in\mathcal{Y}} P(\tilde{y}|\tilde{x},\tilde{s})h\pr{\tilde{f}(\tilde{s},\tilde{x},\tilde{y}),\phi(q,\tilde{y})}\nn\\
    &=\begin{cases} D\left(\left[1,0\right]\big{\|}\left[\frac{2}{3},\frac{1}{3}\right] \right) + h((0,0),1), & \textit{if }\tilde{x}=0, \\
    D\left(\left[1,0\right]\big{\|}\left[\frac{2}{3},\frac{1}{3}\right] \right) + h((0,1),1), & \textit{if } \tilde{x}=1.
    \end{cases}
\end{align}
In both cases of $\tilde{x}$, the equation is simplified to $\log_2\left(\frac{3}{2}\right)$, while the left-hand side of the Bellman equation is $\rho^* + h((0,0),1)$, which is equal to $\log_2\left(\frac{3}{2}\right)$ as well. Therefore, we can conclude that the Bellman equation holds for $(\tilde{s}=(0,0),q=1)$. 

For the MDP state $(\tilde{s}=(0,0), q=2)$, the right-hand side of the Bellman equation is a maximum over $\tilde{x}$ of
\begin{align}
    &D\left(P_{\tilde{Y}|\tilde{X},\tilde{S}}(\cdot|\tilde{x},\tilde{s})\middle\|T_{\tilde{Y}|Q}(\cdot|q)\right)\nn\\&+\sum_{\tilde{y}\in\mathcal{Y}} P(\tilde{y}|\tilde{x},\tilde{s})h\pr{\tilde{f}(\tilde{s},\tilde{x},\tilde{y}),\phi(q,\tilde{y})}\nn\\
    &=\begin{cases} D\left(\left[1,0\right]\big{\|}\left[\frac{1}{3},\frac{2}{3}\right] \right) + h((0,0),1), & \textit{if }\tilde{x}=0, \\
    D\left(\left[1,0\right]\big{\|}\left[\frac{1}{3},\frac{2}{3}\right] \right) + h((0,1),1), & \textit{if } \tilde{x}=1.
    \end{cases}
\end{align}
Also here, in both cases of $\tilde{x}$, the equation is equal to $\log_2\left(3\right)$, while the left-hand side of the Bellman equation is $\rho^* + h((0,0),2) = \log_2(3)$. Thus, the Bellman equation holds for this case too. The verification for the remaining MDP states can be done similarly.

\subsection{Lower Bound}\label{Trapdoor_LB}
The lower bound is derived using Theorem \ref{theorem:bcjr_general_delay} with a particular graph-based encoder that induces the BCJR-invariant property. We show that the achievable rate induced by the graph-based encoder is $R = \log_2\left(\frac{3}{2}\right)$, and therefore $C\geq \log_2\left(\frac{3}{2}\right)$.

A graph-based encoder consists of a $Q$-graph and an input distribution $P_{X|S,Q}$ that is BCJR-invariant. We choose a $Q$-graph consisting of four nodes, and its evolution function is given by the vector representation 
\begin{align}\label{eq:LB_q_vec}
    \underline{\phi}(\underline{q},\tilde{y}=0)&=[1,3,1,3] \nn\\ \underline{\phi}(\underline{q},\tilde{y}=1)&=[2,4,2,4].
\end{align} 
This vector representation implies, for instance, $\phi(q=1,\tilde{y}=0) =1$ and $\phi(q=1,\tilde{y}=1) =2$.
For the $Q$-graph in \eqref{eq:LB_q_vec}, we define the following input distribution:
\begin{align}\label{eq: trapdoor_input_dist}
    &P_{\tilde{X}|\tilde{S},Q}(0|\tilde{s},q) \nn\\&= \scalebox{0.93}{\begin{tabularx}{0.49\textwidth} { 
  |@{}p{1.2cm}@{} 
  | >{\centering\arraybackslash}X 
  | >{\centering\arraybackslash}X 
  | >{\centering\arraybackslash}X 
  | >{\centering\arraybackslash}X | }
 \hline
          & $\tilde{s}=(0,0)$ & $\tilde{s}=(0,1)$ & $\tilde{s}=(1,0)$  & $\tilde{s}=(1,1)$\\
 \hline
$\;\;q=1$  & $2/3$  & $1/3$  & $1/3$ & $0$   \\
\hline
$\;\;q=2$  &  $1$   & $2/3$  & $0$ & $1/3$   \\
\hline
$\;\;q=3$  & $2/3$  & $1$  & $1/3$ & $0$   \\
\hline
$\;\;q=4$  & $1$  & $2/3$  & $2/3$ & $1/3$   \\
\hline
\end{tabularx}}.
\end{align}

According to \eqref{eq: sq_transition}, the Markov transition probability can now be computed as 
\begin{align}
    &P(\tilde{s}^+,q^+|\tilde{s},q)\nn\\
            &= \sum_{\tilde{x},\tilde{y}}P(\tilde{x}|\tilde{s},q)P(\tilde{y}|\tilde{x},\tilde{s})\mathbbm{1}_{\{q^+=\phi(q,\tilde{y})\}}\mathbbm{1}_{{\{\tilde{s}^+=\tilde{f}(\tilde{s},\tilde{x},\tilde{y})\}}}.\nn
\end{align}
Consequently, standard computation of the stationary distribution $\pi(\tilde{s},q)$ provides that
\begin{align*}
   &\pi_{\tilde{S},Q}(\tilde{s},q) \nn\\&= \scalebox{0.93}{\begin{tabularx}{0.49\textwidth} { 
  |@{}p{1.2cm}@{} 
  | >{\centering\arraybackslash}X 
  | >{\centering\arraybackslash}X 
  | >{\centering\arraybackslash}X 
  | >{\centering\arraybackslash}X | }
 \hline
          & $\tilde{s}=(0,0)$ & $\tilde{s}=(0,1)$ & $\tilde{s}=(1,0)$  & $\tilde{s}=(1,1)$\\
 \hline
$\;\;q=1$  & $1/6$  & $1/12$  & $1/36$ & $1/18$   \\
\hline
$\;\;q=2$  & $1/36$  & $1/18$  & $0$ & $1/12$   \\
\hline
$\;\;q=3$  & $1/12$  & $0$  & $1/18$ & $1/36$   \\
\hline
$\;\;q=4$  & $1/18$  & $1/36$  & $1/12$ & $1/6$   \\
\hline
\end{tabularx}}.
\end{align*}

We now verify that the proposed graph-based encoder satisfies the BCJR-invariant property in \eqref{eq: BCJR}. Let us show this explicitly for the case where $(q,\tilde{y})=(1,1)$ and $\tilde{s}^+=(0,0)$. Since $\phi(1,1)=2$, the left-hand side of Eq. \eqref{eq: BCJR} is equal to $\pi_{\tilde{S}|Q}((0,0)|2)$, while the right-hand side is equal to
\begin{align}\label{eq: BCJR_trapdoor}
    &\frac{\sum_{x,s}\pi_{\tilde{S}|Q}(s|1)P_{\tilde{X}|\tilde{S},Q}(x|s,1)P_{\tilde{Y}|\tilde{X},\tilde{S}}(1|x,s)\mathbbm{1}_{\{(0,0)=\tilde{f}(x,1,s)\}}}{\sum_{x',s'}\pi_{\tilde{S}|Q}(s'|1)P_{\tilde{X}|\tilde{S},Q}(x'|s',1)P_{\tilde{Y}|\tilde{X},\tilde{S}}(1|x',s')}\nn\\
    &= \frac{1}{6},\nn
\end{align}
which, indeed, is equal to $\pi_{\tilde{S}|Q}((0,0)|2)$, as required. The verification of the other cases can be done similarly.

Finally, the achievable rate of the graph-based encoder is
\begin{align*}
    R &= I(\tilde{S};\tilde{Y}|Q) \\
      &= \sum_{q\in\mathcal{Q}}\pi_Q(q)\cdot I(\tilde{S};\tilde{Y}|Q=q) \\
      &= \sum_{q\in\mathcal{Q}}\pi_Q(q)\cdot\left[H_2\left(\tilde{Y}|Q=q\right)-H_2(\tilde{Y}|\tilde{S},Q=q)\right]\\
      &\stackrel{(a)}= \sum_{q\in\mathcal{Q}}\pi_Q(q)\cdot\left[H_2\left(\frac{2}{3}\right)-H_2(\tilde{Y}|\tilde{S},Q=q)\right]\\
      &= H_2\left(\frac{2}{3}\right)-\sum_{q\in\mathcal{Q}}\pi_Q(q)\cdot H_2(\tilde{Y}|\tilde{S},Q=q)\\
      &= H_2\left(\frac{2}{3}\right)-\frac{1}{3}\\
      &= \log_2\left(\frac{3}{2}\right),
\end{align*}
where $(a)$ follows due to the fact that
\begin{align}
    P_{\tilde{Y}|Q}(0|q) &= \sum_{\tilde{x},\tilde{s}} \pi(\tilde{s}|q)P(\tilde{x}|\tilde{s},q)P_{\tilde{Y}|\tilde{X}\tilde{S}}(0|\tilde{x},\tilde{s}) \nn\\
    &= \begin{cases} 2/3, & q=1\textit{ or }q=3, \\
                  1/3, & q=2\textit{ or }q=4.\nn
    \end{cases}
\end{align}
\end{proof}

\section{Input-Constrained BSC --- Proof of Theorem \ref{theorem: BSC_UB}}\label{app: BSC}
\begin{proof}
Here we provide the proof of Theorem \ref{theorem: BSC_UB} regarding the upper bound on the capacity of the $(1,\infty)$-input constrained BSC($p$). We begin with the formulation of the channel with delayed feedback of two time instances as a new unifilar FSC with instantaneous feedback. The channel state is defined as $\tilde{s}_{t-1}\triangleq(x_{t-2},x_{t-1})$, the channel input is $\tilde{x}_t=x_t$, and the channel output is $\tilde{y}_t=y_{t-1}$. If $\tilde{s}_{t-1}=(0,0)$ or $\tilde{s}_{t-1}=(1,0)$, then $\tilde{y}_t=0$ with probability $1-p$ or $\tilde{y}_t=1$ with probability $p$. Otherwise, if $\tilde{s}_{t-1}=(0,1)$, then $\tilde{y}_t=0$ with probability $p$ or $\tilde{y}_t=1$ with probability $1-p$. Due to the input constraint, if $\tilde{s}_{t-1}=(0,1)$, then the transmitted input $\tilde{x}_t$ must be zero.

For a particular graph-based test distribution, we solve the MDP problem of the dual capacity upper bound. Here too, consider the $Q$-graph in Eq. \eqref{eq:LB_q_vec}, and the following parameterized graph-based test distribution:
\begin{align}
    T(\tilde{y}=0|\underline{q})=\left[a,b,c,d\right],
\end{align}
where $(a,b,c,d)\in(0,1)^4$. Define the constant
\begin{align}\label{eq: rho_bsc}
    \rho^*&=\log_2\left(\frac{p^p\bar{p}^{\bar{p}}a^{(p^3-3p^2+3p-1)}(\bar{b}\bar{c}d)^{(p^3-p^2)}}{(\bar{a}bc)^{(p^3-2p^2+p)}\bar{d}^{p^3}}\right).
\end{align}
Further, define the value function $h(\tilde{s},q)$ as follows:
\begin{align}\label{eq: value_BSC}
    &h((0,0),1)= h((1,0),1)\nn\\&=\log_2\left(\frac{\bar{c}^pd\bar{d}^pc^{1-2p}}{\bar{a}^{2p}b^pa^{2-3p}}\cdot\left(\frac{\bar{a}bc\bar{d}}{a\bar{b}\bar{c}d}\right)^{p^2}\right)\nn\\
    &h((0,0),2)= h((1,0),2)=\log_2\left(\frac{d}{b}\cdot\left(\frac{b\bar{d}}{\bar{b}d}\right)^p\right)\nn\\
    &h((0,0),3)= h((1,0),3)=\log_2\left(\frac{a^{2p-1}d\bar{d}^p}{(\bar{a}bc)^p}\cdot\left(\frac{\bar{a}bc\bar{d}}{a\bar{b}\bar{c}d}\right)^{p^2}\right)\nn\\
    &h((0,0),4)= h((1,0),4)=0\nn\\
    &h((0,1),1)\nn\\&=\log_2\left(\frac{a^{6p^2-6p+1}\bar{b}^{2p^2-p}\bar{c}^{2p^2}d^{2p^2-p+1}\bar{d}^p}{\bar{a}^{4p^2-2p+1}b^{4p^2-3p+1}c^{4p^2-2p}}\cdot\left(\frac{\bar{a}bc\bar{d}}{a\bar{b}\bar{c}d}\right)^{2p^3}\right)\nn\\
    &h((0,1),2)\nn\\&=\log_2\left(\frac{a^{5p^2-4p+1}(\bar{a}cd)^{p}(\bar{b}\bar{c}d\bar{d})^{p^2}}{\bar{b}^{1-p}(\bar{a}bc)^{3p^2}}\cdot\left(\frac{\bar{a}bc\bar{d}}{a\bar{b}\bar{c}d}\right)^{2p^3}\right)\nn\\
    &h((0,1),3)\nn\\&=\log_2\left(\frac{a^{6p^2-5p+1}\bar{b}^{2p^2-p}\bar{c}^{2p^2+p-1}d^{2p^2-p+1}\bar{d}^p(\bar{a}bc\bar{d})^{2p^3}}{\bar{a}^{4p^2-p}b^{4p^2-3p+1}c^{4p^2-p}(a\bar{b}\bar{c}d)^{2p^3}}\right)\nn\\
    &h((0,1),4)=\log_2\left(\frac{a^{5p^2-4p+1}(\bar{b}\bar{c}d\bar{d})^{p^2}}{(\bar{a}bc)^{3p^2-p}\bar{d}^{1-p}}\cdot\left(\frac{\bar{a}bc\bar{d}}{a\bar{b}\bar{c}d}\right)^{2p^3}\right).
\end{align}

\begin{table*}[b]
\begin{align*}
    T_{Y|Q}(y|q)=
   \scalebox{1}{\begin{tabularx}{0.8\textwidth} { 
  |@{}p{1.3cm}@{} 
  | >{\centering\arraybackslash}X 
  | >{\centering\arraybackslash}X 
  | >{\centering\arraybackslash}X 
  | >{\centering\arraybackslash}X
  | >{\centering\arraybackslash}X
  | >{\centering\arraybackslash}X
  | >{\centering\arraybackslash}X
  | >{\centering\arraybackslash}X | }
 \hline
 & $q=1$ & $q=2$ & $q=3$  & $q=4$  & $q=5$  & $q=6$  & $q=7$  & $q=8$\\
 \hline
$\;\;\;y=-1$  & $0$  & $\gamma_2/2$  & $a/2$ & $a/2$ & $\gamma_2/2$ & $0.5-\gamma_1$ & $\gamma_2/2$ & $a/2$   \\
\hline
$\;\;\;y=0$  & $\gamma_1$  & $0.5-\gamma_2$  & $0.5-a$ & $0.5-a$ & $0.5-\gamma_2$ & $\gamma_1$ & $0.5-\gamma_2$ & $0.5-a$   \\
\hline
$\;\;\;y=1$  & $0.5-\gamma_1$  & $\gamma_2/2$  & $a/2$ & $a/2$ & $\gamma_2/2$ & $0$ & $\gamma_2/2$ & $a/2$   \\
\hline
$\;\;\;y=?$  & $0.5$  & $0.5$  & $0.5$ & $0.5$ & $0.5$ & $0.5$ & $0.5$ & $0.5$   \\
\hline
\end{tabularx}}.
\end{align*}
\caption{Graph-based test distribution for the DEC.}
\label{table: test_dist}
\end{table*}
To complete the proof it is left to show that, under the constraints given in \eqref{eq: bsc_cons}, the scalar $\rho^*$ in \eqref{eq: rho_bsc} and the value function in \eqref{eq: value_BSC} solve the Bellman equation. For the MDP state $(\tilde{s}=(0,0),q=1)$, the right-hand side of the Bellman equation is a maximum over $\tilde{x}$ of
\begin{align}\label{eq: BSC_bellman_rs}
    &D\left(P_{\tilde{Y}|\tilde{X},\tilde{S}}(\cdot|\tilde{x},\tilde{s})\middle\|T_{\tilde{Y}|Q}(\cdot|q)\right)\nn\\&+\sum_{\tilde{y}\in\mathcal{Y}} P(\tilde{y}|\tilde{x},\tilde{s})h\pr{\tilde{f}(\tilde{s},\tilde{x},\tilde{y}),\phi(q,\tilde{y})}\nn\\
    &=D\left(\left[\bar{p},p\right]\big{\|}\left[a,\bar{a}\right] \right)\nn\\ 
    &+\begin{cases} 
    \bar{p}\cdot h((0,0),1)+p\cdot h((0,0),2), & \textit{if } \tilde{x}=0\\
    \bar{p}\cdot h((0,1),1)+p \cdot h((0,1),2), & \textit{if } \tilde{x}=1.
    \end{cases}
\end{align}

Under the constraints given in \eqref{eq: bsc_cons}, it can be verified that $\tilde{x}=0$ attains the maximum in \eqref{eq: BSC_bellman_rs}. Further, the left-hand side of the Bellman equation is $\rho^*+h(\left(0,0),1\right)$, which, after being simplified, is exactly equal to the right-hand side of the Bellman equation. Hence, the Bellman equation holds for the case that $(\tilde{s}=(0,0),q=1)$. The verification for the remaining MDP states is omitted here and follows similar calculations.
\end{proof}

\section{DEC --- Proof of Theorem \ref{theorem: DEC_UB}}\label{app: DEC}
\begin{proof}
Here we provide the proof of Theorem \ref{theorem: DEC_UB} regarding the upper bound on the capacity of the DEC. As before, we start with the formulation of the channel with delayed feedback of two time instances as a new unifilar FSC with instantaneous feedback. The channel state is defined as $\tilde{s}_{t-1}\triangleq(x_{t-2},x_{t-1})$, the channel input is $\tilde{x}_t=x_t$, and the channel output is $\tilde{y}_t=y_{t-1}$. The output of the DEC is $\tilde{y}_t=x_{t-1}-x_{t-2}$ with probability $1-p$, or $\tilde{y}_t=?$ with probability $p$, where $p \in [0,1]$ is the channel parameter. 

Also here, for a particular graph-based test distribution, we solve the MDP problem of the dual capacity upper bound. Specifically, consider the following $Q$-graph:
\begin{align}\label{eq:UB_q_vec_DEC}
    \underline{\phi}(\underline{q},\tilde{y}=-1)&=[1,1,1,1,1,1,1,1] \nn\\ \underline{\phi}(\underline{q},\tilde{y}=0)&=[1,3,3,4,4,6,8,8] \nn\\
    \underline{\phi}(\underline{q},\tilde{y}=1)&=[6,6,6,6,6,6,6,6] \nn\\
    \underline{\phi}(\underline{q},\tilde{y}=?)&=[2,7,7,7,7,5,7,7].
\end{align}
For $a,\gamma_1,\gamma_2\in(0,0.5)$ and the $Q$-graph in \eqref{eq:UB_q_vec_DEC}, consider the graph-based test distribution $T_{Y|Q}(y|q)$ that is given by Table \ref{table: test_dist}.
The proposed graph-based test distribution follows by first numerically optimizing over the test distribution. Then, we observed that the optimal test distribution can be represented by three parameters, which are denoted here as $\gamma_1,\gamma_2$, and $a$. 

Define the constant
\begin{align}\label{eq: rho_DEC}
    &\rho^*=\frac{1}{4}\log_2\left(\frac{(1-2a)^{-1}(2-3a)}{\left(1+8a^2\bar{a}-3a-(1-4a\bar{a})\sqrt{1+4a^3}\right)}\right).
\end{align}
Also, define $h(\tilde{s},q)$ as follows:
\begin{align}\label{eq: value_DEC}
    &h((0,0),1)= h((0,0),6)=\frac{1}{4}\log_2\left(\frac{a(1-2a)\gamma_2}{4(1-2\gamma_2)\gamma_1^2}\right)\nn\\
    &h((0,0),2)= h((0,0),5)= h((0,0),7)\nn\\&=\frac{1}{4}\log_2\left(\frac{4a\gamma_1^2\gamma_2}{(1-2a)(1-2\gamma_2)^3}\right)\nn\\
    &h((0,0),3)= h((0,0),4)= h((0,0),8)\nn\\&=\frac{1}{4}\log_2\left(\frac{4a\gamma_1^2\gamma_2}{(1-2a)^3(1-2\gamma_2)}\right)\nn\\
    &h((0,1),1)= \frac{1}{4}\log_2\left(\frac{a(1-2a^2)}{(1-2\gamma_1)^3}\right)\nn\\
    &h((0,1),2)= h((0,1),5)= h((0,1),7)\nn\\&=\frac{1}{4}\log_2\left(\frac{a(1-2a)^2}{\gamma_2^2(1-2\gamma_1)}\right)\nn\\
    &h((0,1),3)= h((0,1),4)= h((0,1),8)\nn\\&=\frac{1}{4}\log_2\left(\frac{(1-2a)^2}{a(1-2\gamma_1)}\right)\nn\\
    &h((0,1),6)= \frac{1}{4}\log_2\left(\frac{a(1-2a)^2}{1-2\gamma_1}\right)\nn\\
    &h((1,0),1)= \frac{1}{4}\log_2\left(\frac{a\gamma_2(1-2a)}{1-2\gamma_2}\right)\nn\\
    &h((1,0),2)= h((1,0),5)= h((1,0),7)\nn\\&=\frac{1}{4}\log_2\left(\frac{a(1-2a)}{\gamma_2(1-2\gamma_2)}\right)\nn\\
    &h((1,0),3)= h((1,0),4)= h((1,0),8)\nn\\&=\frac{1}{4}\log_2\left(\frac{a(1-2\gamma_2)}{\gamma_2(1-2a)}\right)\nn\\
    &h((1,0),6)= \frac{1}{4}\log_2\left(\frac{a\gamma_2(1-2a)}{(1-2\gamma_1)^2(1-2\gamma_2)}\right)\nn\\
    &h((1,1),1)=\log_2\left(\frac{1-2a}{2\gamma_1}\right)\nn\\
    &h((1,1),2)= h((1,1),5)= h((1,1),7)= \frac{1}{2}\log_2\left(\frac{1-2a}{1-2\gamma_2}\right)\nn\\
    &h((1,1),3)= h((1,1),4)= h((1,1),8)= 0 \nn\\
    &h((1,1),6)= \frac{1}{4}\log_2\left(\frac{a(1-2a)^2}{4\gamma_1^2(1-2\gamma_1)}\right).
\end{align}
Let us assume that the optimal policy is given by 
\begin{align}\label{eq: policy_DEC}
    u^*(\tilde{s},q)=\begin{cases} 1, & (\tilde{s}=(1,1),q=1), \\
                  0, & \textit{else}.
    \end{cases}
\end{align}
The policy above was obtained by solving the MDP problem numerically using the value iteration algorithm. Assuming \eqref{eq: policy_DEC}, it can be noted that the Bellman equation is based on a finite set of linear equations, and it can be verified that if
\begin{align}\label{eq:param_dec}
    \gamma_1 &= \frac{1}{4a}\left((2-4a)\cdot\sqrt{a^2+0.25}+4a\bar{a}-1\right)\nn\\
    \gamma_2 &= \frac{1}{4-6a}\left((4a^2-4a+1)\sqrt{1+4a^2}-8a^2\bar{a}+1\right),
\end{align}
then the Bellman equation holds under our choice of $\rho^*$ in \eqref{eq: rho_DEC} and the function $h(\tilde{s},q)$ in \eqref{eq: value_DEC}. The verification follows from straightforward calculations, as we did in the previous sections, and therefore the details are omitted here. In Eq. \eqref{eq:param_dec}, we write analytical expressions for $\gamma_1, \gamma_2$ as a function of $a$. These expressions were derived by observing that, for particular MDP states, the optimal solution (in terms of the test distribution's parameters) is achieved when the right-hand side of the Bellman equation does not depend on the action. Namely, for particular MDP states we require that the right-hand side of the Bellman equation is equal for $u=0$ and $u=1$. Such a requirement results in linear equality constraints that are satisfied with $\gamma_1$ and $\gamma_2$ in \eqref{eq:param_dec}.

\end{proof}

\end{appendices}
\bibliography{ref}
\bibliographystyle{IEEEtran}
\newpage
\begin{IEEEbiographynophoto}{Bashar Huleihel}
(Student Member, IEEE) received the B.Sc. and M.Sc. degrees in electrical and computer engineering from the Ben-Gurion University of the Negev, Israel, in 2017 and 2020, respectively. He is currently pursuing the Ph.D. degree in electrical and computer engineering at the same institution. His research interests include information theory and machine learning.
\end{IEEEbiographynophoto}

\begin{IEEEbiographynophoto}{Oron Sabag}
(Member, IEEE) received his B.Sc. (cum laude), the M.Sc. (summa cum laude) and the Ph.D. in Electrical and Computer Engineering from the Ben-Gurion University of the Negev, Israel, in 2013, 2016 and 2019, respectively. From 2019 to 2022, he was a postdoctoral fellow with the Department of Electrical Engineering at Caltech. He is currently a senior lecturer at the Benin School of Computer Science and Engineering at The Hebrew University of Jerusalem. His research interests include control, information theory and reinforcement learning. He was a recipient of several awards, among them are ISEF postdoctoral fellowship, ISIT-2017 best student paper award, SPCOM-2016 best student paper award, and the Feder Family Award for outstanding research in communications.
\end{IEEEbiographynophoto}

\begin{IEEEbiographynophoto}{Haim Permuter}
(Senior Member, IEEE) received the B.Sc. and M.Sc. degrees (summa cum laude) in electrical and computer engineering from Ben-Gurion University of the Negev, Israel, in 1997 and 2003, respectively, and the Ph.D. degree in electrical engineering from Stanford University, Stanford, CA, USA, in 2008. From 1997 to 2004, he was an Officer with the Research and Development Unit of the Israeli Defense Forces. Since 2009, he has been with the Department of Electrical and Computer Engineering, Ben-Gurion University of the Negev, where he is currently a Professor and the Luck-Hille Chair of electrical engineering. He also serves as the Head of the communication track in his department. He was a recipient of several awards, among them the Fullbright Fellowship, the Stanford Graduate Fellowship (SGF), the Allon Fellowship, and the U.S.–Israel Binational Science Foundation Bergmann Memorial Award. He has served on the editorial boards for the IEEE TRANSACTIONS ON INFORMATION THEORY from 2013 to 2016.
\end{IEEEbiographynophoto}

\begin{IEEEbiographynophoto}{Victoria Kostina}(Senior Member, IEEE)
is a professor of electrical engineering and of computing and mathematical sciences at Caltech. She received the bachelor's degree from Moscow Institute of Physics and Technology (MIPT) in 2004, the master's degree from University of Ottawa in 2006, and the Ph.D. degree from Princeton University in 2013.  During her studies at MIPT, she was affiliated with the Institute for Information Transmission Problems of the Russian Academy of Sciences. 
 
Her research program spans finite delay theory of information, random access communications, control over communication channels, and information bottlenecks in computing systems. She has served as a Guest Editor for the IEEE Journal on Selected Areas in Information Theory and for Entropy.  She received the Natural Sciences and Engineering Research Council of Canada postgraduate scholarship during 2009--2012, Princeton Electrical Engineering Best Dissertation Award in 2013, Simons-Berkeley research fellowship in 2015 and the NSF CAREER award in 2017.  
\end{IEEEbiographynophoto}

\end{document}